\newcommand{\CLASSINPUTtoptextmargin}{1.4in}
\newcommand{\CLASSINPUTbottomtextmargin}{1.4in}
\documentclass[11pt, draftclsnofoot, onecolumn,article,romanappendices]{IEEEtran}
\usepackage{setspace}
\usepackage{cite}
\ifCLASSINFOpdf
\usepackage[pdftex]{graphicx}
\else
\usepackage[dvips]{graphicx}
\fi

\usepackage{subfigure}
\usepackage{hyperref}
\usepackage[cmex10]{amsmath}
\usepackage{amssymb}
\usepackage{amsfonts}
\usepackage{amsthm}
\usepackage{array}
\usepackage{dsfont}
\usepackage{epstopdf}

\newtheorem{lemma}{Lemma}
\newtheorem{theorem}{Theorem}

\newtheorem{corollary}{Corollary}

\newtheorem{proposition}{Proposition}

\long\def\symbolfootnote[#1]#2{\begingroup%
\def\thefootnote{\fnsymbol{footnote}}\footnote[#1]{#2}\endgroup}

\begin{document}

\title{Distortion Exponent in MIMO Fading Channels with Time-Varying Source Side Information}

%
\author{
  \IEEEauthorblockN{I\~naki Estella Aguerri and Deniz G\"{u}nd\"{u}z\footnote{
  I\~naki Estella Aguerri is with the Mathematical and Algorithmic Sciences Lab, France Research Center, Huawei Technologies Co. Ltd., Boulogne-Billancourt, France. E-mail: {\tt inaki.estella@huawei.com \tt}.
  Deniz G\"{u}nd\"{u}z is with the Department of Electrical and Electronic Engineer at Imperial College London, London, UK. E-mail: {\tt  d.gunduz}{\tt @imperial.ac.uk}\\
    Part of the reseach was done during the Ph.d. studies of I\~naki Estella Aguerri at Imperial College London.  This paper was presented in part at the  IEEE International Conference on
Communications, Kyoto, Japan, Jun. 2011  \cite{estella2011icc}, at the IEEE International Symposium
on Information Theory, St. Petersburg, Russia, Aug. 2011 \cite{ Estella2011DistExponent}, and at the  
IEEE Global Conference on Signal and Information Processing (GlobalSIP),Texas, U.S., Dec.  2013\cite{Estella2013DistSideInfo}
.}\\}
}

\maketitle

\vspace{-1.5cm}
\begin{abstract}
Transmission of a Gaussian source over a time-varying multiple-input multiple-output (MIMO) channel is studied under strict delay constraints. Availability of  a correlated side information at the receiver is assumed, whose quality, i.e., correlation with the source signal, also varies over time. A block-fading model is considered for the states of the time-varying channel and the time-varying side information; and perfect
state information at the receiver is assumed, while the transmitter
knows only the statistics. The high SNR performance, characterized by the \textit{distortion exponent}, is studied for this joint source-channel coding problem. An upper bound is derived and compared with lowers based on list decoding, hybrid digital-analog transmission, as well as  multi-layer schemes which transmit successive refinements of the source, relying on progressive and superposed transmission with list decoding.
The optimal distortion exponent is characterized for the single-input multiple-output (SIMO) and multiple-input single-output (MISO) scenarios  by showing that the distortion exponent achieved by multi-layer superpositon encoding with joint decoding meets the proposed upper bound. In the MIMO scenario, the optimal distortion exponent is characterized in the low bandwidth ratio regime, and it is shown that the multi-layer superposition encoding performs very close to the upper bound in the high bandwidth expansion regime.
\end{abstract}

\vspace{-0.5cm}


\IEEEpeerreviewmaketitle
\section{Introduction}

Many applications in wireless networks require the transmission of a
source signal over a fading channel, i.e., multimedia
signals over cellular networks or the accumulation of local
measurements at a fusion center in sensor networks, to be reconstructed with the
minimum distortion possible at the destination. In many
practical scenarios, the destination receives additional correlated
side information about the underlaying source signal, either form other transmitters in the network, or
through its own sensing devices. For example, measurements from
other sensors at a fusion center, signals from repeaters in digital
TV broadcasting, or relay signals in mobile networks.

The theoretical benefits of having  correlated side information
at the receiver for source encoding are well known \cite{wyner1978rate}. However, similar to estimating the channel state information at the transmitter, it is costly to provide an estimate of the available side information to the transmitter, or may even be impossible in uncoordinated scenarios. Without the knowledge of the channel and the side information states, a transmitter needs to transmit in a manner that can adapt dynamically to the time-varying channel and side information qualities without knowing their realizations.

Here, we consider the joint source-channel coding problem of transmitting a Gaussian source over a multiple-input multiple-output (MIMO) block-fading channel when the receiver has access to time-varying correlated source side information. Both the time-varying channel and the source side-information are assumed to follow block-fading models, whose states are unknown at the transmitter. Moreover, strict delay constraints apply requiring the transmission of a block of source samples, for which the side-information state is constant, over a block of the channel, during which the channel state is also constant. The source and channel blocks do not necessarily have the same length, and their ratio is defined as the \emph{bandwidth ratio} between the channel and the source bandwidths.

We are interested in minimizing the average end-to-end distortion of the reconstructed source samples, averaged over many blocks. This may correspond to the average distortion over video frames in a video streaming application, where each frame has to be transmitted under a strict delay constraint.

When the knowledge of the channel and side information states is available at both the transmitter and the receiver (CSI-TR), Shannon's separation theorem applies \cite{Shamai:IT:98}, assuming that the channel and source blocks are sufficiently long.  However, the optimality of separation does not extend to non-ergodic scenarios such as the model studied in this paper, since each source block is required to be transmitted over a single channel block.
We note that the suboptimality of separate source and channel coding is dependent on the performance criterion under study. For example, it was shown in \cite{Peng2010:DistOut} that, if, instead of the average distortion, the outage distortion is considered, separate source and channel coding is still optimal.

This problem has been studied extensively in the literature in the absence of correlated side information at the receiver \cite{ng2007minimumLayered, Etemadi2006OptLaye, Alnuweiri2012UtilityMaxBC}. Despite the ongoing efforts, the minimum achievable average distortion remains an open problem; however, more conclusive results on the performance can be obtained by studying the \emph{distortion exponent}, which characterizes the exponential decay of the expected distortion in the high SNR regime \cite{Laneman:2004}. The distortion exponent has been studied for parallel fading channels in \cite{Gunduz2006DistExpParall}, for the relay channel in \cite{Gunduz2007RelayChann}, for point-to-point MIMO channels in \cite{gunduz2008joint}, for channels with feedback in \cite{Gunduz2009DistExpFeedback}, for the two-way relay channel in \cite{Wang2011DistExpTwoWayRelay}, for the interference channel in \cite{Zhao2010Interference}, and in the presence of side information that might be absent in \cite{Zhao2010SideInfoGaussian}. In the absence of source side information at the receiver, the optimal distortion exponent in MIMO channels is known in some regimes of operation, such as the large bandwidth regime \cite{gunduz2008joint} and the low bandwidth regime \cite{Caire2007hybrid}. However, the general problem remains open. In \cite{gunduz2008joint} successive refinement source coding followed by superposition transmission is shown to achieve the optimal distortion exponent for high bandwidth ratios in MIMO systems. The optimal distortion exponent in the low bandwidth  ration regime is achieved through hybrid digital-analog transmission \cite{gunduz2008joint, Caire2007hybrid}. In \cite{bhattad2008distortion}, superposition multi-layer schemes are
shown to achieve the optimal distortion exponent for some other bandwidth ratios as well. 

The source coding version of our problem, in which the encoder and decoder are connected by an error-free finite-capacity link, is studied in \cite{ng2007minimum}. The single-input single-output (SISO) model in the presence of a time-varying channel and side information is considered for matched bandwidth ratios in \cite{Estella2013Systematic}, where uncoded transmission is shown to achieve the minimum expected distortion for certain side information fading gain distributions, while separate source and channel coding is shown to be suboptimal in general. A scheme based on list decoding at the receiver, is also proposed in \cite{Estella2013Systematic}, and it is shown to outperform separate source and channel coding by exploiting the joint quality of the channel and side information states.

Our goal in this work is to find tight bounds on the distortion exponent when transmitting a Gaussian source over a time-varying MIMO channel in the presence of time-varying correlated source side information at the receiver\footnote{Preliminary results have been published in the conference version of this work in \cite{estella2011icc} for SISO channels and in \cite{Estella2011DistExponent} and \cite{Estella2013DistSideInfo} for MIMO channels.}. 

The main results of this work can be summarized as follows:
\begin{itemize}
\item We derive an upper bound on the distortion exponent by providing the channel state realization to the transmitter, while the source side information state remains unknown.
\item We characterize the distortion exponent achieved by the list decoding (LD) scheme. While this scheme achieves a lower expected distortion than SSCC, we show that it does not improve the distortion exponent.
\item Based on LD, we consider a hybrid digital-analog list decoding scheme (HDA-LD) and extend LD by considering multi-layer transmission, where each layer carries successive refinement information for the source sequence. We consider both the progressive (LS-LD) and superposition (BS-LD) transmission of these layers,  and derive the respective distortion exponents.
\item We show that the distortion exponent achieved by BS-LD meets the proposed upper bound for SISO/SIMO/MISO systems, thus characterizing the optimal distortion exponent in these scenarios. We show that HDA-LD also achieves the optimal distortion exponent in SISO channels.
\item In the general MIMO setup, we characterize the optimal distortion exponent in the low bandwidth ratio regime, and show that it is achievable by both HDA-LD and BS-LD. In addition, we show that, in certain regimes of operation, LS-LD outperforms all the other proposed schemes.
\end{itemize}

We will use the following notation in the rest of the paper. We denote random variables with upper-case letters, e.g., $X$, their realizations with lower-case letters, e.g., $x$, and the sets with calligraphic letters, e.g. $\mathcal{A}$. We denote $\mathrm{E}_{X}[\cdot]$ as the expectation
with respect to $X$, and $\mathrm{E}_{\mathcal{A}}[\cdot]$ as the expectation
over the set $\mathcal{A}$. We denote random vectors as $\mathbf{X}$ with realizations $\mathbf{x}$.  We denote by $\mathds{R}^+$ the set of positive real numbers, and by $\mathds{R}^{++}$ the set of strictly positive real numbers in $\mathds{R}$, respectively. We define $(x)^+=\max\{0, \nu \}$. Given two functions $f(x)$ and $g(x)$, we use $f(x)\doteq g(x)$ to denote the exponential equality $\lim_{x\rightarrow\infty}\frac{\log f(x)}{\log g(x)}=1$, while  $\stackrel{.}{\geq}$ and $\stackrel{.}{\leq}$ are defined similarly.

The rest of the paper is organized as follows. The problem statement is given in Section \ref{sec:ProblemStatement}. Two upper bounds on the distortion exponent are derived in Section \ref{sec:UpperBounds}. Various achievable schemes are
studied in Section \ref{sec:SingleLayer}. The characterization of the optimal distortion exponent for certain regimes is relegated to Section \ref{sec:Comments}. Finally, the conclusions are presented in Section \ref{sec:Conclusions}.


\section{Problem Statement}\label{sec:ProblemStatement}

\begin{figure}
\centering
\includegraphics[width=0.65\textwidth]{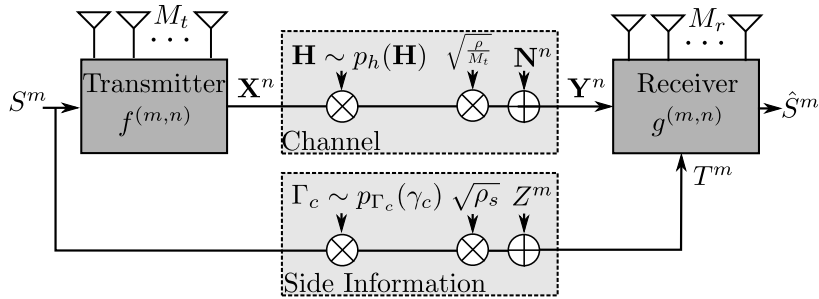}
\vspace{-3 mm} \caption{Block diagram of the joint source-channel
coding problem with fading channel and side information qualities.}
\label{fig:Model}
\end{figure}

We wish to transmit a zero mean, unit variance complex Gaussian
source sequence $S^m \in \mathds{C}^m$ of independent and
identically distributed (i.i.d.) random variables, i.e., $S_i\sim
\mathcal{CN}(0,1)$, over a complex MIMO block Rayleigh-fading channel with
$M_t$ transmit and $M_r$ receiver antennas, as shown in Figure
\ref{fig:Model}. In addition to the channel output, time-varying correlated source side information is also available at the receiver. Time-variations in the source side information are assumed to follow a block fading model as well. The channel and the side
information states are assumed to be constant for the duration of
one block, and independent of each other, and among different blocks. We assume that
each source block is composed of $m$ source samples, which, due to
the delay limitations of the underlying application, must
be transmitted over one block of the channel, which consists of $n$
channel uses. We define the \emph{bandwidth ratio} of the system as
\begin{IEEEeqnarray}{rCl}
 b\triangleq\frac{n}{m}\quad \textit{channel dimension per  source sample}. 
\end{IEEEeqnarray}

The encoder maps each source sequence $S^{m}$ to a channel input
sequence $\mathbf{X}^n=[\mathbf{X}_1,...,\mathbf{X}_n]\in \mathds{C}^{M_t\times n}$ using an encoding
function $f^{(m,n)}:\mathds{C}^m\rightarrow\mathds{C}^{M_t\times n}$ such that the
average power constraint is satisfied: $\sum_{i=1}^n\mathrm{Tr}\{\mathrm{E}[\mathbf{X}_i^{H}
\mathbf{X}_i]\}\leq n\cdot M_t $. If codeword $\mathbf{x}^n$ is transmitted, the signal received at the destination is modeled by memoryless slow fading channel 
\begin{equation}
\mathbf{Y}_i = \sqrt{\frac{\rho}{M_t}}\boldsymbol{\mathsf{H}} \mathbf{x}_i +
\mathbf{N}_i,\qquad i=1,...,n,
\end{equation}
where $\boldsymbol{\mathsf{H}}\in \mathds{C}^{M_r\times M_t}$ is the channel matrix with  i.i.d. zero mean complex Gaussian entries, i.e., $H_{ij}\sim\mathcal{CN}(0,1)$, whose realizations are denoted by $\mathbf{H}$, $\rho\in \mathds{R}^+$ is the average signal to noise ratio (SNR) in the channel, and $\mathbf{N}_i$
models the additive noise with
$\mathbf{N}_i\sim\mathcal{CN}(0,\mathbf{I})$.  We define $M^*
=\max\{M_t,M_r\}$ and $M_*=\min\{M_t, M_r\}$, and consider  $\lambda_{M_*}\geq\cdots\geq\lambda_1>0$ to be the eigenvalues of $\boldsymbol{\mathbf{H}}\boldsymbol{\mathbf{H}}^H$.

 In addition to the channel output
$\mathbf{Y}^n=[\mathbf{Y}_1,...,\mathbf{Y}_{n}]\in\mathds{C}^{M_r\times n}$, the decoder observes $T^m\in\mathds{C}^m$, a randomly degraded version of the source sequence:
\begin{IEEEeqnarray}{rCl}
T^m=\sqrt{\rho_s} \Gamma_c  S^m + Z^m,
\end{IEEEeqnarray}
where $\Gamma_c\sim\mathcal{CN}(0,1)$ models time-varying Rayleigh fading
in the quality of the side information, $\rho_s\in \mathds{R}^+$ models the average quality of the side information, and $Z_j\sim\mathcal{N}(0,1)$, $j=1,...,m$, models the noise. We define the \emph{side information gain} as $\Gamma\triangleq |\Gamma_c|^2$, and its realization as $\gamma$. Then, $\Gamma$ follows an exponential distribution with probability density function (pdf):
\begin{IEEEeqnarray}{rCl}
p_{\Gamma}(\gamma)=e^{-\gamma},\qquad \gamma\geq0.
\end{IEEEeqnarray}

In this work, we assume that the receiver knows the side information and channel realizations, $\gamma$ and $\mathbf{H}$, while the transmitter is only aware of their distributions. The receiver reconstructs the source
sequence $\hat{S}^m=g^{(m,n)}(\mathbf{Y}^n,T^m,\mathbf{H},\gamma)$ with a
mapping $g^{(m,n)}:\mathds{C}^{n\times M_r} \times
\mathds{C}^m\times \mathds{C}^{M_t\times M_r}\times \mathds{R}
\rightarrow \mathds{C}^m$. The distortion between the source sequence and the reconstruction is
measured by the quadratic average
distortion $D\triangleq\frac{1}{m}\sum^m_{i=1}|S_i-\hat{S}_i|^2$.

 We are interested in characterizing the minimum
\emph{expected distortion}, $\mathrm{E}[D]$, where the expectation is taken
with respect to the source, the side information and channels state realizations, as well as the noise terms, and expressed as
\begin{IEEEeqnarray}{rCl}
ED^*(\rho,\rho_s,b)\triangleq\lim_{
n,m\rightarrow\infty}\max_{n \leq m b}
\min_{f^{(m,n)},g^{(m,n)}}\mathrm{E}[D].
\end{IEEEeqnarray}

In particular, we are interested in characterizing  the optimal performance in the high SNR regime, i.e., when $\rho,
\rho_s \rightarrow \infty$. We define $\nu$ as a measure of the average \emph{side information quality} in the high SNR regime, as follows:
\begin{IEEEeqnarray}{rCl}
 \nu\triangleq\lim_{\rho\rightarrow\infty}
 \frac{\log \rho_s}{\log \rho}.
 \end{IEEEeqnarray}

Parameter $\nu$ captures the increase in the quality of the side information with respect to the average SNR in the channel. For example, if the side information is made available to the receiver through other transmissions, if the average SNR in the channel increases, so does the side information quality.

The performance measure we consider is the
\emph{distortion exponent}, defined as
\begin{align}
\Delta(b,\nu)\triangleq-\lim_{
\substack{\rho,
\rho_s \rightarrow \infty\\\rho_s \doteq \rho^{\nu}}
}\frac{\log{\mathrm{E}[D]}}{\log{\rho}}.
\end{align}

\section{Distortion Exponent Upper Bound}\label{sec:UpperBounds}
In this section we derive an upper bound on the distortion exponent by extending the bound on the expected distortion $ED^*$ obtained in \cite{Estella2013Systematic} to the MIMO setup with bandwidth mismatch, and analyzing the high SNR behavior.
The upper bound is constructed by providing the transmitter with only  the channel state realization, $\mathbf{H}$, while the side information state, $\gamma$, remains unknown. We call this the \emph{partially informed encoder upper bound}. The optimality of separate source and channel coding is shown in \cite{Estella2013Systematic} when the side information fading gain distribution is discrete, or continuous and quasiconcave for $b=1$. The proof easily extends to the non-matched bandwidth ratio setup and, since in our model $p_{\Gamma}(\gamma)$ is exponential, and hence, is continuous and quasiconcave, separation is optimal at each channel block\footnote{Although in our setup the side information state $\Gamma_c$ is complex, the receiver can always correct the phase to have an equivalent real side information state with Rayleigh amplitude. Our setup then reduces to the two parallel problems of reconstructing  the real and imaginary parts of $S^m$ with the same side information gain, and all the techniques of \cite{Estella2013Systematic} can be applied. }.

As shown in \cite{ng2007minimum, Estella2013Systematic}, if $p_{\Gamma}(\gamma)$ is monotonically decreasing, the optimal source encoder ignores the side information completely, and the side-information is used only at the decoder for source reconstruction\footnote{ We note that when the distribution of the side information is not Rayleigh, the optimal encoder follows a different strategy. For example, for quasiconcave continuous distributions the optimal source code compresses the source aiming at a single target side information state. See \cite{Estella2013Systematic} for
details.}. Concatenating this side-information-ignorant source code with a channel code at the instantaneous capacity, the minimum expected distortion at each channel state $\mathbf{H}$ is given by
\begin{align}
D_{\text{op}}(\rho,\rho_s,b,\mathbf{H})=
\frac{1}{\rho_s}e^{\frac{2^{b\mathcal{C}(\mathbf{H})}}{\rho_s}}E_{1}\left(\frac{2^{b\mathcal{C}(\mathbf{H})}}{\rho_s}\right),
\end{align}
where $E_1(x)$ is the exponential integral given by
$E_1(x)=\int_{x}^{\infty}t^{-1}e^{t}dt$.
Averaging over the channel state realizations, the expected distortion is lower bounded as
\begin{IEEEeqnarray}{rCl}\label{eq:LowerBound}
ED_{\text{pi}}^*(\rho,\rho_s,b)&=&\mathrm{E}_{\boldsymbol{\mathsf{H}}}[D_{\text{op}}(\rho,\rho_s,b,\boldsymbol{\mathsf{H}})].
\end{IEEEeqnarray}

Then, an upper bound on the distortion exponent is found by analyzing the high SNR behavior of (\ref{eq:LowerBound}). This upper bound will be expressed in terms of the diversity-multiplexing tradeoff (DMT), which measures the tradeoff between the rate and reliability in the transmission of a message over a MIMO fading channel in the asymptotic high SNR regime \cite{zheng2003diversity}.
For a family of channel codes with rate $R=r\log \rho$, where $r$ is the \textit{multiplexing gain}, the 
 DMT is the piecewise-linear function $d^*(r)$ connecting the points $(k,d^*(k))$, $k=0,...,M_*$, where $d^*(k)=(M^*-k)(M_*-k)$. More specifically, for $r\geq M_*$, we have $d^*(r)=0$, and for $0\leq r\leq M_*$ satisfying $k\leq r\leq k+1$ for some $k=0,1,..., M_*-1$, the DMT curve is characterized by
\begin{IEEEeqnarray}{rCl}\label{eq:DMTcurve}
d^*(r)&\triangleq& \Phi_k-\Upsilon_k(r-k),
\end{IEEEeqnarray}
where we have defined
\begin{IEEEeqnarray}{rCl}\label{eq:ML_PhiUpsilon}
\Phi_k &\triangleq& (M^* - k) (M_* - k) \quad\text{and }\quad
\Upsilon_k \triangleq (M^* + M_* - 2 k - 1 ).
\end{IEEEeqnarray}

\begin{theorem} \label{the:UpperBound}
Let $l=1$ if  $\nu/M_*< M^*-M_*+1$, and let $l\in\{2,...,M_*\}$ be the integer satisfying $2l-3+M^*-M_*\leq \nu/M_*<2l-1+M^*-M_*$ if $M^*-M_*+1\leq \nu/M_*< M^*+M_*-1$. The distortion exponent is
upper bounded by
\begin{IEEEeqnarray}{lCl}
\Delta_{up}(b,\nu)=\begin{cases}
\nu&\text{if }0\leq b<\frac{ \nu}{M_*},\\
bM_*&\text{if }\frac{\nu}{M_*}\leq b< M^*-M_*+1,\\
\nu+d^*\left(\frac{\nu}{b}\right)&\text{if }
1+M^*-M_*\leq b<2l-1+M^*-M_*,\\
\nu+d^*\left(\frac{\nu}{b}\right)&\text{if }
2l-1+M^*-M_*\leq b<\frac{\nu}{M_*-k},\\
\Delta_{\mathrm{MIMO}}(b)&\text{if
}\frac{\nu}{M_*-k}\leq b<M^*+M_*-1,\\
\nu+d^{*}\left(\frac{\nu}{b}\right)&\text{if
} M^*+M_*-1\leq b,
\end{cases}
\end{IEEEeqnarray}
where $k\in\{l,..., M_*-1\}$ is the integer satisfying $2k-1+M^*-M_*\leq b <2k+1+M^*-M_*$, and
\begin{IEEEeqnarray}{rCl}\label{eq:MIMOUpperBound}
\Delta_{\mathrm{MIMO}}(b)\triangleq\sum_{i=1}^{M_*}\min\{b,2i-1+M^*-M_*\}.
\end{IEEEeqnarray}

 If $\nu/M_*\geq M^*+M_*-1$, then
\begin{IEEEeqnarray}{lCl}
\Delta_{up}(b,\nu)=\nu+d^*\left(\frac{\nu}{b}\right),
\end{IEEEeqnarray}
where $d^*(r)$ is the DMT characterized in (\ref{eq:DMTcurve})-(\ref{eq:ML_PhiUpsilon}).
\end{theorem}

\begin{proof}
The proof is given in Appendix  \ref{app:UpperBound}.
\end{proof}


A looser bound, denoted as the \emph{fully informed encoder upper bound}, is obtained if both the channel state $\mathbf{H}$ and the side information state $\gamma$ are provided to the transmitter. At each realization, the problem reduces to the static setup studied in \cite{Shamai:IT:98}, and source-channel separation theorem applies; that is, the concatenation of a Wyner-Ziv source code with a capacity achieving channel code is optimal at each realization. Analyzing its high SNR behavior following similar derivations in \cite{gunduz2008joint} and \cite{Estella2013Systematic}, we have that, the distortion exponent is upper bounded by
\begin{IEEEeqnarray}{rCl}\label{eq:InformedUpperBound}
\Delta_{\mathrm{inf}}(b,\nu)=\nu+\Delta_{\mathrm{MIMO}}(b)
\end{IEEEeqnarray}

Comparing the two upper bounds in (\ref{eq:InformedUpperBound}) and Theorem \ref{the:UpperBound}, we can see that the latter is always tighter. When $\nu>0$, the two bounds meet only at the two extremes, when either $b=0$ or $b\rightarrow \infty$. Note that these bounds provide the achievable distortion exponents when either both states (Equation (\ref{eq:InformedUpperBound})) or only the channel state (Theorem  \ref{the:UpperBound}) is available at the transmitter, also characterizing the potential gains from channel state feedback in fading joint source-channel coding problems. 

\section{Achievable Schemes}\label{sec:SingleLayer}
In this section, we propose transmission schemes consisting of a single-layer and multi-layer codes, and analyze their achievable distortion exponent performances.

\subsection{List decoding scheme (LD)} \label{ssec:VBD}

SSCC is optimal in the presence of CSI-TR. However, when CSI-TR is not available, the binning and channel coding rates have to be designed  based only on the statistics. As shown in \cite{Estella2013Systematic}, transmission using SSCC suffers from two separate outage events: outage in channel decoding and outage in source decoding. It is shown in \cite[Corollary 1]{Estella2013Systematic} that, for monotonically decreasing pdfs, such as $p_{\Gamma}(\gamma)$ considered here, the expected distortion is minimized by avoiding outage in source decoding, that is, by not using binning. Then, the optimal SSCC scheme compresses the source sequence at a fixed rate ignoring the source side information, and transmits the compressed bits over
the channel using a channel code at a fixed rate. At the receiver, first the transmitted channel codewor is recoverd. If the channel decoding is successful, the compression codeword is recovered, and the source sequence is reconstructed together with the side information. Otherwise, only the side information is used for reconstruction.

Instead of using explicit binning at the source encoder, and decoding a single channel codeword, in LD the channel decoder outputs a list of channel codeword candidates, which are then used by the source decoder together with the source side information to find the transmitted source codeword. The success of decoding for this scheme depends on the joint quality of the channel and side information states. This scheme is considered in \cite{Estella2013Systematic} for a SISO system, and is shown to outperform SSCC at any SNR, and to achieve the optimal distortion exponent in certain regimes, while SSCC remains suboptimal.

At the encoder, we generate a codebook of $2^{mR_{ld}}$ length-$m$
quantization codewords $W^m(i)$
through a `test channel' given by $W=S+Q$, where
$Q\sim\mathcal{CN}(0,\sigma_Q^2)$ is independent of $S$; and an independent Gaussian codebook of size
$2^{nbR_{ld}}$ with length-$n$ codewords
$\mathbf{X}(i)\in\mathds{C}^{M_t\times n}$, where $\mathbf{X}\sim
\mathcal{CN}(0,\mathbf{I})$, such that $bR_{ld}=I(S;W)+\epsilon$, for an arbitrarily small $\epsilon>0$, i.e.,  with $\sigma^2_Q=(2^{bR_{ld}-\epsilon}-1)^{-1}$. Given a source outcome
$S^m$, the transmitter finds the quantization codeword $W^m(i)$
jointly typical with the source outcome, and transmits the
corresponding channel codeword $\mathbf{X}(i)$. The channel decoder looks for the list of indices $\mathcal{I}$ of jointly typical codewords $(\mathbf{X}^n(l),\mathbf{Y}^n)$. Then, the source decoder, finds the unique $W^n(i)$ jointly typical with $T^m$ among the codewords $W^n(l)$, $l\in \mathcal{I}$. 

Joint decoding produces a binning-like decoding: only some $\mathbf{Y}^n$  are
jointly typical with $\mathbf{X}(i)$, generating a virtual bin, or list, of
$W^m$ codewords from which only one is jointly
typical with $T^m$ with high probability.
 The size of the list depends on the
realizations of $\boldsymbol{\mathsf{H}}$ and $\Gamma$ unlike in a Wyner-Ziv scheme, in which the bin sizes are chosen in advance. Therefore, the outage event depends jointly on the channel and the side
information states $(\mathbf{H},\gamma)$. 
An outage is declared whenever, due to the channel and side information randomness, a unique codeword cannot be recovered, and is given by
\begin{IEEEeqnarray}{rCl}\label{eq:OutageSets}
\mathcal{O}_{ld}&=&\left\{(\mathbf{H},
\gamma):I(S;W|T)\geq bI(\mathbf{X};\mathbf{Y})\right\},
\end{IEEEeqnarray}
where $I(\mathbf{X};\mathbf{Y})=\log\det(\mathbf{I}+\frac{\rho}{M_*}\mathbf{HH}^H)$ and $I(S;W|T)=\log(1+(2^{bR_{ld}-\epsilon}-1)/(\gamma\rho_s+1))$.

If $W^m$  is successfully decoded, the source sequence is estimated with an MMSE
estimator using the quantization codeword and the side information sequence, i.e., $\hat{S}_i=\mathrm{E}[S_i|W_i,T_i]$, and reconstructed with a distortion $D_d(bR_{ld},\gamma)$, where
\begin{IEEEeqnarray}{rCL}\label{eq:DistSSCC}
D_d(R,\gamma)\triangleq(\rho_{s}\gamma+2^{R})^{-1}.
\end{IEEEeqnarray}

If there is an outage, only the side information is
used in source reconstruction, and the corresponding distortion is given by $D_d(0,\gamma)$.
 Then, the
expected distortion for LD is expressed as
\begin{IEEEeqnarray}{rCl}\label{eq:Dj}
ED_{ld}(R_{j})&=&\mathrm{E}_{\mathcal{O}^c_{ld}}\left[D_d\left(bR_{ld},\Gamma\right)\right]+
\mathrm{E}_{\mathcal{O}_{ld}}[D_d(0,\Gamma)].
\end{IEEEeqnarray}

\begin{theorem}\label{the:ExponentJoint}
The achievable distortion exponent for LD, $\Delta_{ld}(b,\nu)$, is given by
\begin{IEEEeqnarray}{rCl}\label{eq:ExponentJoint}
&\Delta_{ld}&(b,\nu)= \max\left\{\nu,
b\frac{\Phi_k+k\Upsilon_k+\nu}{\Upsilon_k+b}\right\},\qquad\text{for } b\in\left[\frac{\Phi_{k+1}+\nu}{k+1},
\frac{\Phi_k+\nu}{k}\right), k=0,1,...,M_*-1,
\end{IEEEeqnarray}
where $\Phi_k$ and $\Upsilon_k$ are as defined in (\ref{eq:ML_PhiUpsilon}).
\end{theorem}
\begin{proof}
See Appendix \ref{app:ExponentJoint}.
\end{proof}

LD reduces the probability of outage, and hence, the expected distortion compared to SSCC. Figure \ref{fig:LDSISOMIMO} shows the expected distortion achievable by  SSCC and LD schemes, as well as the partially informed encoder lower bound on the expected distortion in a SISO and a $3\times 3$ MIMO system for $b=2$. It is observed that LD outperforms SSCC in both SISO and MIMO scenarios, although both schemes fall short of the expected distortion lower bound, $ED_{\mathrm{pi}}^*$. We also observe that both schemes keep a constant performance gap as the SNR increases. In fact, the next proposition, given without proof, reveals that both schemes achieve the same distortion exponent.
\begin{proposition}
The distortion exponent of LD, $\Delta_{ld}(b,\nu)$,  is the same as that SSCC, i.e., $\Delta_{ld}(b,\nu)=\Delta_{s}(b,\nu)$.
\end{proposition}

We note that although LD and SSCC achieve the same distortion exponent in the current setting, LD is shown to achieve larger distortion exponents than SSCC in general\cite{Estella2013Systematic}.

\begin{figure}
\centering
\includegraphics[width=0.65\textwidth]{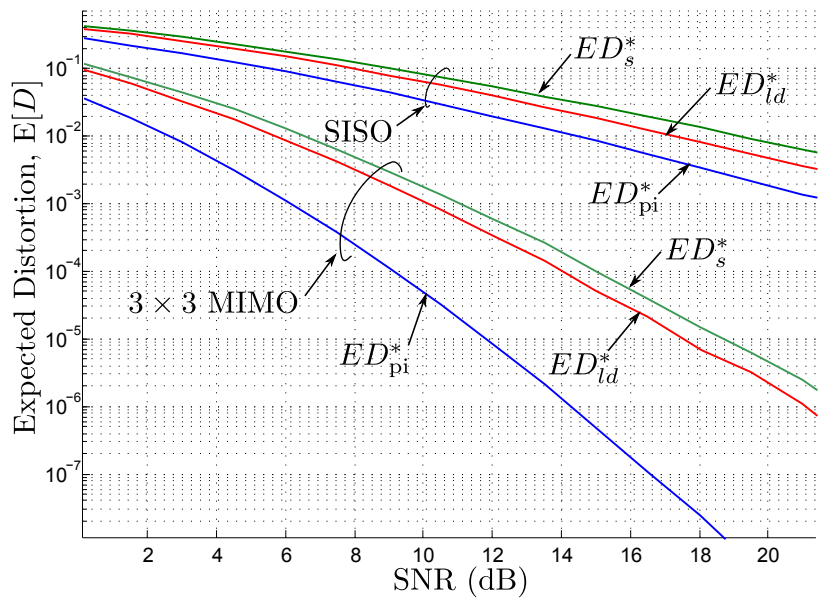}
\caption{Minimum expected distortion achievable by SSCC and LD for a SISO and a $3\times3$ MIMO channel for $b=2$ and $\nu=1$. The partially informed encoder bound is also included.}\label{fig:LDSISOMIMO}
\end{figure}

 Next, we extend the idea of list decoding as a building block for more advanced transmission strategies.

\subsection{Hybrid digital-analog list decoding  scheme (HDA-LD)}\label{sec:HDA}

We introduce a hybrid digital-analog (HDA) scheme that quantizes the source sequence, uses a
scaled version of the quantization error as the channel input, and exploits list decoding at the decoder. This scheme is introduced in \cite{wilson2010joint}, and shown to be optimal in static SISO channels in the presence of side information for $b=1$. HDA-LD is considered in \cite{Estella2013Systematic} in the SISO fading setup with $b=1$, and is shown to achieve the optimal distortion exponent for a wide family of side information distributions. In this paper, we propose a generalization of HDA-LD in \cite{Estella2013Systematic} to the MIMO channel and to bandwidth ratios satisfying $b\geq1/M^*$.

For $ b\leq 1/M_*$, we ignore the available side information and use the hybrid digital-analog scheme  proposed in \cite{Caire2007hybrid}. In this scheme, which we denote by superposed HDA (HDA-S), the source sequence is transmitted in two layers. The first layer transmits a part of the source sequence in an uncoded fashion, while the second layer digitally transmits the remaining samples. The two layers are superposed and the power is allocated between the two. At the receiver, the digital layer is decoded treating the uncoded layer as noise. Then, the source sequence is reconstructed using both layers. The distortion exponent achievable by HDA-S is given by $\Delta_{h}(b,\nu)=bM_*$ for $b\leq1/ M_*$ \cite{Caire2007hybrid}.

HDA-S can be modified to include list decoding and to use the available side information at the reconstruction to reduce the expected distortion. However, as we will show in Section \ref{sec:LowBand}, if $0\leq b\leq \nu/M_*$, simple MMSE estimation of the source sequence is sufficient to achieve the optimal distortion exponent, given by $\Delta^*(b,\nu)=\nu$, while, if $\nu/M_*\leq b\leq 1/M_*$, HDA-S is sufficient to achieve the optimal distortion exponent. Therefore, HDA-S with list decoding does not improve the distortion exponent in this regime.
 \begin{lemma}
 The distortion exponent achievable by HDA-S is given by $\Delta_{h}(b,\nu)=bM_*$, if $b\leq 1/M_*$.
\end{lemma}
Next, we consider the HDA-LD scheme for $bM_*>1$.
At the encoder, a quantization codebook of $2^{mR_h}$ length-$m$ codewords $W^m(s)$, $s=1,...,2^{m R_{h}}$, is generated with a test channel $W=S+Q$, where $Q\sim\mathcal{CN}(0,\sigma^2_Q)$ is independent of $S$, and the quantization noise variance is chosen such that $R_{h}=I(W;S)+\epsilon$, for an arbitrarily small $\epsilon>0$, i.e., $\sigma_Q^2\triangleq(2^{R_{h}-\epsilon}-1)^{-1}$. Then, each $W^m$ is reordered into length-$\frac{m}{M_*}$ codewords
$\mathbf{W}(s)=[\mathbf{W}_1(s),...,\mathbf{W}_{\frac{m}{M_*}}(s)]\in\mathds{C}^{\frac{m}{M_*}\times M_*}$, where $\mathbf{W}_i(s)$, $i=1,...,m/M_*$, is given by $\mathbf{W}_i(s)=[W_{(i-1)M_*+1}(s);...;W_{iM_*}(s)]^T$. Similarly, we can reorder $S^m$ and $Q^m$, and define $\mathbf{S}_i$ and $\mathbf{Q}_i$ for $i=1,...,m/M_*$.

We then generate $2^{mR_{h}}$ independent auxiliary random vectors
$\mathbf{U} \in \mathds{C}^{\left(n-\frac{m}{M_*}\right)\times
M_*}$ distributed as
$\mathbf{U}_i\sim\mathcal{CN}(0,\mathbf{I})$, for
$i=1,...,n-\frac{m}{M_*}$ and assign one to each $\mathbf{W}(s)$ to construct the codebook of size $2^{mR_{h}}$ consisting of the pairs of codewords $(\mathbf{W}(s),\mathbf{U}(s))$, $s=1,...,2^{m R_{h}}$.  For a given
source sequence $S^m$, the encoder looks for the $s^*$-th codeword
$\mathbf{W}(s^*)$ such that $(\mathbf{W}(s^*),S^m)$ are jointly
typical. A unique $s^*$ is found if $M_*R_{h}>I(\mathbf{W};\mathbf{S})$. Then, the pair $(\mathbf{W}(s^*),\mathbf{U}(s^*))$ is used
to generate the channel input, which is scaled to satisfy the power
constraint:
\begin{IEEEeqnarray}{rCl}
\mathbf{X}_i=\begin{cases}
\sqrt{\frac{1}{\sigma_Q^2}}[\mathbf{S}_i-\mathbf{W}_i(s^*)], & \text{for } i=1,...,\frac{m}{M_*},\\
\mathbf{U}_{i-\frac{m}{M_*}}(s^*), &
\text{for } i=\frac{m}{M^*}+1,...,n.
\end{cases}
\end{IEEEeqnarray}
Basically, in the first block of $\frac{m}{M_*}$ channel accesses we transmit a scaled version of the error of the quantization codeword $\mathbf{Q}_i$ in an uncoded fashion, while in the second block of $n-\frac{m}{M_*}$ accesses we transmit a digital codeword.


At the receiver, list-decoding is
successful with high probability if \cite{Estella2013Systematic}
\begin{IEEEeqnarray}{rCl}\label{eq:DecCond}
I(\mathbf{W};\mathbf{S})<M_*R_h <I(\mathbf{W}\mathbf{U};\mathbf{Y} \mathbf{T})
\end{IEEEeqnarray}
The outage event can be shown to be given, after some algebra, as follows.
\begin{IEEEeqnarray}{rCl}\label{eq:OutageRegionHDASec}
\mathcal{O}_{h}=\Bigg{\{}(\mathbf{H},\gamma): I(\mathbf{W},\mathbf{S})
&\geq& I(\mathbf{W};\mathbf{Y}_W \mathbf{T})
+ (bM_*-1)I(\mathbf{U};\mathbf{Y}_{U})\Bigg{\}},
\end{IEEEeqnarray}
where $I(\mathbf{U};\mathbf{Y}_U)=\log\det(\mathbf{I}+\frac{\rho}{M_*}\mathbf{HH}^H)$ and,
\begin{IEEEeqnarray}{rCl}\label{eq:MItermHDA}
I(\mathbf{W};\mathbf{Y}_W \mathbf{T})=\log\left(\frac{(\xi(1+\sigma_Q^2))^{M_*}\det(\mathbf{I}+\frac{\rho}{M_t}\mathbf{HH}^H))}{\det(\mathbf{I}+\sigma^2_Q(\frac{\rho}{M_t}\mathbf{HH}^H+\xi\mathbf{I}))}\right),
\end{IEEEeqnarray}
where $\xi\triangleq1+\rho_s\gamma$.

If $W^{m}$ is successfully  decoded, each $X^n$ is
reconstructed with an MMSE estimator using $Y^n$ and $T^m$ within a distortion given by
\begin{IEEEeqnarray}{rCl}\label{eq:HDADist}
D_{h}(\sigma_Q^2,\mathbf{H},\gamma)\!=\!\frac{1}{M_*}\sum_{i=1}^{M_*}\left(
1+\rho_s\gamma+\frac{1}{\sigma^2_Q}\left(1+\frac{\rho}{M_*}\lambda_i\right)\huge\right)^{-1}.
\end{IEEEeqnarray}

If an outage occurs and $W^m$ is not decoded, only $T^m$ is used in the reconstruction, since $\mathbf{X}^n$ is uncorrelated with the source sequence by construction, and so is $\mathbf{Y}^n$. Using an MMSE estimator, the achievable distortion is given by $D_{d}(0,\gamma)$. Then, the expected distortion for 
HDA-LD is given by
\begin{IEEEeqnarray}{rCl}
ED_h(R_h)=\mathrm{E}_{\mathcal{O}^c_h}[D_{h}(\sigma_Q^2,\boldsymbol{\mathsf{H}},\Gamma)]+\mathrm{E}_{\mathcal{O}_h}[D_{d}(0,\Gamma)].
\end{IEEEeqnarray}

The distortion exponent of HDA-LD, $\Delta_{h}(b,\nu)$, is characterized in the next theorem.
\begin{theorem}\label{th:DistortionExponentHDAExponent}
Let $bM_*>1$. The distortion exponent achieved by HDA-LD, $\Delta_{h}(b,\nu)$, is given by 

\begin{IEEEeqnarray}{lCl}
\Delta_{h}(b,\nu)=\begin{cases}
\nu&\text{if }
\frac{1}{M_*}\leq b<\frac{\nu}{M_*},\\
1+\frac{(bM_*-1)
(\Phi_k+k\Upsilon_k-1+\nu)}{bM_*-1+M_*\Upsilon_k},&\text{if }
b\in
\left[\frac{\Phi_{k+1}-1+\nu}{k+1}+\frac{1}{M_*},\frac{\Phi_k-1+\nu}{k}+\frac{1}{M_*}\right),\\
& \quad\text{for }k=0,...,M_*-1.
\end{cases}
\end{IEEEeqnarray}

\end{theorem}
\begin{proof}
See Appendix \ref{app:DistortionExponentHDAExponent}.
\end{proof}

\subsection{Progressive multi-layer LD transmission (LS-LD)}\label{sec:ProgressiveLD}

In this section, we consider a multi-layer transmission scheme to improve the achievable distortion exponent, in particular in the high bandwidth ratio regime. Multi-layer transmission is proposed in \cite{gunduz2008joint} to combat channel fading by transmitting multiple layers that carry successive refinements of the source sequence. At the receiver, as many layers as possible are decoded depending on the channel state. The better the channel state, the more layers can be decoded and the smaller is the distortion at the receiver. We propose to use successive refinement codewords that exploit the side information at the receiver \cite{Steinberg2004SuccessiveRefiWZ}.
Then, the refinement codewords are transmitted one after the other over the channel using the LD scheme introduced in Section \ref{ssec:VBD}. Similarly to \cite{gunduz2008joint}, we assume that each layer is allocated the same time resources (or number of channel accesses). In the limit of infinite layers, this assumption does not incur a loss in performance.

At the encoder, we generate $L$ Gaussian quantization codebooks, each with $2^{mR_l}$ codewords $W_l^m$ and $bR_l/L=I(S;W_l|W_1^{l-1})+\epsilon$, for $l=1,...,L$, with an arbitrarily small $\epsilon>0$, such that each Gaussian
codebook is a refinement for the previous layers  \cite{Steinberg2004SuccessiveRefiWZ}. The quantization codewords $W_l^n$ are generated with a test channel given by $W_l=S+\sum_{i=l}^L
Q_i$, for $l=1,...,L$, where $Q_l \sim
\mathcal{N}(0,\sigma^2_l)$ are independent of $S$ and of each other. Note that $T-S-W_L-W_{L-1}-\cdots-W_1$ form a Markov chain.
For a given rate tuple $\mathbf{R}\triangleq [R_1,...,R_L]$, with $R_1\geq\cdots R_L\geq 0$, the quantization noise variances satisfy
\begin{IEEEeqnarray}{rCl}\label{eq:Vark}
\sum_{i=l}^L\sigma_i^2=(2^{\sum_{i=1}^l(\frac{b}{L}R_i-\epsilon)}-1)^{-1},
\qquad l=1,...,L.
\end{IEEEeqnarray}

We generate $L$ independent channel codebooks, each with $2^{n\frac{bR_l}{L}}$ length-$\frac{n}{L}$  codewords $\mathbf{X}^{n/L}_{l}\in\mathds{C}^{M_t\times n/L}$ with $\mathbf{X}_{l,i}\sim\mathcal{CN}(0,\mathbf{I})$. Each successive refinement codeword is transmitted using LD as in Section \ref{ssec:VBD}.
At the destination, the decoder successively decodes each refinement codeword using joint decoding from the first layer up to the $L$-th layer. Then, $l$ layers will be successfully decoded if
\begin{IEEEeqnarray}{rCl}
I(S;W_{l}|T,W_{1}^{l-1})<\frac{b}{L} I(\mathbf{X};\mathbf{Y})\leq I(S;W_{l+1}|T,W^{l}_1),
 \end{IEEEeqnarray}
that is, $l$ layers are successfully decoded while there is an outage in decoding the $(l+1)$-th layer. Let us define the outage event, for $l=1,...,L$, as follows
\begin{IEEEeqnarray}{rCl}\label{eq:ProgLDOutCond}
\mathcal{O}^{ls}_{l}\!\triangleq\!\left\{(\mathbf{H},\gamma)\!:\!I(S;W_{l}|T,W_{1}^{l-1})\geq\frac{b}{L}I(\mathbf{X};\mathbf{Y}) \right\},
\end{IEEEeqnarray}
where $I(\mathbf{X},\mathbf{Y})=\log\det\left(\mathbf{I}+\frac{\rho}{M_*}\mathbf{HH}^{H}\right)$, and, with $R_0\triangleq0$,
\begin{IEEEeqnarray}{rCl}\label{eq:RHS1}
I(S;W_l|W_1^{l-1},T)=\log\left(\frac{2^{\sum_{i=1}^l\frac{b}{L}R_i}+\gamma\rho_s}{2^{\sum_{i=1}^{l-1}\frac{b}{L}R_i}+\gamma\rho_s}\right).
\end{IEEEeqnarray}
The details of the derivation are given in Appendix \ref{app:ProgressiveLD}.
Due to the successive refinability of the Gaussian source even in the presence of side information \cite{Steinberg2004SuccessiveRefiWZ}, provided $l$ layers have been successfully decoded, the receiver reconstructs the source with an MMSE estimator using the side information and the decoded layers with a distortion given by $D_{d}(\sum_{i=1}^lbR_l/L,\gamma)$.
The expected distortion can be expressed as follows.
\begin{IEEEeqnarray}{rCl}\label{eq:EDLS}
ED_{ls}(\mathbf{R})=\sum_{l=0}^{L}\mathrm{E}_{(\mathcal{O}^{ls}_{l})^c\bigcap\mathcal{O}^{ls}_{l+1}}\left[D_{d}\left(\sum_{i=1}^{l}\frac{bR_l}{L},\gamma\right)\right]\!.
\end{IEEEeqnarray}

The distortion exponent achieved by LS-LD is given next.

\begin{theorem}\label{th:DistExpLSLD}
Let us define
\begin{IEEEeqnarray}{rCl}\label{eq:LS_alpha_and_M}
\phi_k&\triangleq& M^*-M_*+2 k -1,\quad M_k\triangleq M_*-k+1,
\end{IEEEeqnarray}
and the sequence $\{c_i\}$  as
\begin{IEEEeqnarray}{rCl}
c_0=0, \quad c_i=c_{i-1}+\phi_i\ln\left(\frac{M_i}{M_i-1}\right),
\end{IEEEeqnarray}
for $i=1,...,M_*-1$, and $c_{M_*}=\infty$.

 The distortion exponent achieved by LS-LD with infinite number of layers is given by $\Delta^*_{ls}(b,\nu)=\nu$ if $b\leq x/M_*$, and if
\begin{IEEEeqnarray}{rCl}
c_{k-1}+\frac{\nu}{M_k}< b\leq c_{k}+\frac{\nu}{M_k-1},
\end{IEEEeqnarray}
for some $k\in\{1,...,M_*\}$, the achievable distortion exponent is given by
\begin{IEEEeqnarray}{rCl}
\Delta^*_{ls}(b,\nu)
&=&\nu+\sum_{i=1}^{k-1}\phi_i+M_k\phi_k
\times\left(1-e^{-\frac{b(1-\kappa^*)-c_{k-1}}{\phi_k}}\right),
\end{IEEEeqnarray}
where
\begin{IEEEeqnarray}{rCl}
\kappa^*=\frac{\phi_k}{b}\mathcal{W}\left(\frac{e^{\frac{b-c_{k-1}}{\phi_k}}\nu}{M_k \phi_k}\right),
\end{IEEEeqnarray}
and $\mathcal{W}(z)$ is the Lambert $W$ function, which gives the principal solution for $w$ in $z=we^{w}$.
\end{theorem}
\begin{proof}
See Appendix \ref{app:ProgressiveLD}.
\end{proof}

The proof of Theorem \ref{th:DistExpLSLD} indicates that the distortion exponent for LS-LD is achieved by allocating an equal rate among the first $\kappa^*L$ layers to guarantee that the distortion exponent is at least $\nu$. Then, the rest of the refinement layers are used to further increase the distortion exponent with the corresponding rate allocation. Note that for $\nu=0$, we have $\kappa^*=0$, and Theorem \ref{th:DistExpLSLD} boils down to Theorem 4.2 in \cite{gunduz2008joint}.

\subsection{Broadcast strategy with LD (BS-LD)}\label{sec:MultiLD}
In this section, we consider the broadcast strategy in which the successive refinement layers are transmitted by superposition, and are decoded one by one with list decoding. The receiver decodes as many layers as possible using successive joint decoding, and reconstructs the source sequence using the  successfully decoded layers and the side information sequence.

At the encoder, we generate $L$ Gaussian quantization codebooks, at
rates $bR_l=I(S;W_l|W_1^{l-1})+\epsilon$, $l=1,...,L$, $\epsilon>0$, as in Section \ref{sec:ProgressiveLD}, and  $L$ channel codebooks $\mathbf{X}^n_{l}$,
$l=1,\ldots,L$, i.i.d. with $\mathbf{X}_{l,i}\sim\mathcal{CN}(0,\mathbf{I})$. Let $\boldsymbol\rho=[\rho_1,...,\rho_L,\rho_{L+1}]^T$ be the power allocation
among channel codebooks such that $\rho=\sum_{i=1}^{L+1}\rho_i$.  We
consider a power allocation strategy, such that $\rho_l=\rho^{\xi_{l-1}}-\rho^{\xi_{l}}$ with
$1=\xi_0\geq\xi_1\geq\ldots \geq \xi_L \geq 0$, and define $\boldsymbol\xi\triangleq[\xi_1,...,\xi_L]$. In the last layer, the layer $L+1$, Gaussian i.i.d. noise sequence with distribution $\mathbf{\tilde{N}}_i\sim\mathcal{CN}(0,\mathbf{I})$ is transmitted using  the remaining power $\rho_{L+1}\triangleq \rho^{\xi_L}$ for mathematical convenience. Then, the channel input $\mathbf{X}^n$ is generated as the superposition
of the $L$ codewords, $\mathbf{X}^n_{l}$ with the corresponding power allocation
$\sqrt{\rho_l}$ as
\begin{IEEEeqnarray}{rCL}
\mathbf{X}^n=\frac{1}{\sqrt{\rho}}\sum_{j=1}^L\sqrt{\rho_j}\mathbf{X}^n_{j}+\sqrt{\rho^{\xi_{L}}}\mathbf{\tilde{N}}^n.
\end{IEEEeqnarray}

 At the receiver, successive joint decoding is used from layer $1$ up to layer $L$, considering the posterior layers as noise. Layer $L+1$, containing the noise, is ignored. The outage event at layer $l$, provided $l-1$ layers have been decoded successfully, is given by
\begin{IEEEeqnarray}{rCl}\label{eq:OutageJointMultiple}
\mathcal{O}^{bs}_{l}&=&\left\{(\mathbf{H},\gamma):b
I(\mathbf{X}_{l};\mathbf{Y}|\mathbf{X}_{1}^{l-1})\leq I(S;W_l|T,W_1^{l-1})\right\}.
\end{IEEEeqnarray}
If $l$ layers are decoded,  the source is reconstructed at a distortion $D_{d}(\sum_{i=1}^{l}bR_i,\gamma)$ with an MMSE estimator, and the expected
distortion is found as
\begin{IEEEeqnarray}{rCl}
ED_{bs}(\mathbf{R},\boldsymbol\xi)\!=\!\sum_{l=1}^L\mathrm{E}_{\mathcal{O}^{bs}_{l+1}}
\left[D_d\left(\sum^l_{i=0}bR_i,\Gamma\right)\right],
\end{IEEEeqnarray}
where $\mathbf{R}\triangleq [R_1,...,R_L]$ and $\mathcal{O}^{bs}_{L+1}$ is the set of states in which  all the $L$ layers are successfully decoded.

The problem of optimizing the distortion exponent for BS-LD for $L$ layers, which we denote by $\Delta_{bs}^{L}(b,\nu)$, can be formulated as a linear program over the multiplexing gains $\mathbf{r}\triangleq [r_1,...,r_l]$, where $R_l=r_l\log\rho$ for $l=1,...,L$, and the power allocation $\boldsymbol\xi$, as shown in (\ref{eq:MlSystem2}) in Appendix  \ref{app:ExponentLayerLD_DistExpSol}, and can be efficiently solved numerically. In general, the performance of BS-LD is improved by increasing the number of layers $L$, and an upper bound on the performance, denoted by $\Delta_{bs}^*(b,\nu)$, is given in the limit of infinite layers, i.e., $L\rightarrow\infty$, which can be approximated by numerically solving $\Delta_{bs}^{L}(b,\nu)$ with a large number of layers. However, obtaining a complete  analytical characterization of $\Delta_{bs}^{L}(x,b)$ and $\Delta_{bs}^*(b,\nu)$ in general is complicated. In the following, we fix the multiplexing gains, and optimize the distortion exponent over the power allocation.  While fixing the multiplexing gains is potentially suboptimal, we obtain a closed form expression for an achievable distortion exponent, and analytically evaluate its limiting behavior. We shall see that, as the number of layers increases, this analytical solution matches the numerically optimized distortion exponent.

First, we fix the multiplexing gains as $\hat{\mathbf{r}}=[\hat{r}_1,...,\hat{r}_L]$ where $\hat{r}_l=[(k+1)(\xi_{l-1}-\xi_{l})-\epsilon_1]$ for $l=1,...,L$, for some $\epsilon_1\rightarrow0$, and optimize the distortion exponent over $\boldsymbol\xi$. The achievable distortion exponent is given in the next theorem.

\begin{theorem}\label{the:MLfintieLayers}
Let us define
\begin{IEEEeqnarray}{rCl}\label{eq:ML_eta}
 \eta_k\triangleq\frac{b(k+1)-\Phi_{k+1}}{\Upsilon_k}\quad\text{and } \quad \Gamma_k\triangleq\frac{1-\eta_{k}^{L-1}}{1-\eta_k}.
\end{IEEEeqnarray}

The distortion exponent achievable by BS-LD with $L$ layers and multiplexing gain $\hat{\mathbf{r}}$, is given by $\hat{\Delta}^L_{bs}(b,\nu)=\nu$ for $bM_*\leq  \nu$, and by
\begin{IEEEeqnarray}{lCl}\label{ML_distExpL}
\hat{\Delta}^L_{bs}(b,\nu)=\nu+\Phi_k-\frac{\Upsilon_k(\Upsilon_k(\nu+\Phi_k)+\nu b(k+1)\Gamma_k)}{(\Upsilon_k+b(1+k))(\Upsilon_k+b(1+k)\Gamma_k)-b(k+1)\Phi_k\Gamma_k},
\end{IEEEeqnarray}
for
\begin{IEEEeqnarray}{rCl}
b&\in&\left[\frac{\Phi_{k+1}+\nu}{k+1},\frac{\Phi_{k}+\nu}{k} \right), \quad k=0,...,M_*-1.
\end{IEEEeqnarray}
\end{theorem}

\begin{proof}
See Appendix \ref{app:ExponentLayerLD_DistExpSol}.
\end{proof}

An upper bound on the performance of BS-LD with multiplexing gains $\hat{\mathbf{r}}_l$ is
obtained for a continuum of infinite layers, i.e.,
$L\rightarrow\infty$.
\begin{corollary}\label{cor:MlInfLay}
The distortion exponent of BS-LD with multiplexing gains $\hat{\mathbf{r}}$ in the limit of
infinite layers, $\hat{\Delta}_{bs}^{\infty}(b,\nu)$, is given, for $k\!=\!0,...,M_*\!-\!1$, by
\begin{IEEEeqnarray}{lCl}
\hat{\Delta}_{bs}^{\infty}(b,\nu)=\max\{\nu,b(k+1)\}\qquad
\text{for  } b\!\in\!\left[\frac{\Phi_{k+1}+\nu}{k+1},\frac{\Phi_k}{k+1}\right),
\end{IEEEeqnarray}
and \begin{IEEEeqnarray}{lCl}
\hat{\Delta}_{bs}^{\infty}(b,\nu)&=&\Phi_k+\nu\left(\frac{b (1+k)-\Phi_k}{b(1+k)-\Phi_{k+1}}\right)\qquad\text{for }b\in\left[\frac{\Phi_k}{k+1},\frac{\Phi_k+\nu}{k}\right).
\end{IEEEeqnarray}
\end{corollary}
\begin{proof}
See Appendix \ref{app:ExponentLayerLD_DistExpSol}.
\end{proof}

 The solution in Theorem \ref{the:MLfintieLayers} is obtained by fixing the multiplexing gains to $\hat{\mathbf{r}}$. This is potentially suboptimal since it excludes, for example, the performance of single-layer LD from the set of feasible solutions. By fixing $\mathbf{r}$ such that $r_2=\cdots=r_L=0$, BS-LD reduces to single layer LD and achieves a distortion exponent given in Theorem \ref{the:ExponentJoint}, i.e., $\Delta_{ld}(b,\nu)$. Interestingly, for $b$ satisfying
\begin{IEEEeqnarray}{rCl}\label{eq:CondBDreduced}
b\in\left[\frac{\Phi_k}{k},\frac{\Phi_k+\nu}{k}\right),\qquad k=1,...,M_*-1,
\end{IEEEeqnarray}
single-layer LD achieves a larger distortion exponent than $\hat{\Delta}_{bs}^{\infty}(b,\nu)$ in Corollary \ref{cor:MlInfLay}, as shown in Figure \ref{fig:MIMODistortionExponentMultiBad}. Note that this region is empty for $\nu=0$, and thus, this phenomena does not appear in the absence of side information. The achievable distortion exponent for BS-LD can be stated as follows.

\begin{lemma}\label{lem:BSLDAchievable}
BS-LD achieves the distortion exponent
\begin{IEEEeqnarray}{rCl}
\bar{\Delta}_{bs}(b,\nu)&=&\max\{\hat{\Delta}_{bs}^{\infty}(b,\nu),\Delta_{ld}(b,\nu)\}.
\end{IEEEeqnarray}
\end{lemma}

Next, we consider the numerical optimization $\Delta_{bs}^L(b,\nu)$, and compare it with the distortion exponent achieved by fixing the multiplexing gain.
In Figure \ref{fig:MIMODistortionExponentMultiBad} we show one instance of the numerical optimization of $\Delta^{L}_{bs}(b,s)$ for $3\times 2$ MIMO and $\nu=0.5$, for $L=2$ and $L=500$ layers. We also include the distortion exponent achievable by single-layer LD, i.e., when $L=1$, and the exponent achievable by BS-LD with multiplexing gains $\hat{\mathbf{r}}$, with $L=2$ layers and in the limit of infinite layers, denoted by $\hat{\Delta}_{bs}^{2}(b,\nu)$ and $\hat{\Delta}_{bs}^{\infty}(b,\nu)$, respectively. We observe that the numerically optimized distortion exponent improves as the number of layers increases. There is a significant improvement in the distortion exponent just by using two layers in the high bandwidth regime, while this improvement is not so significant for intermediate $b$ values. We also note that there is a tight match between the distortion exponent achievable by Lemma \ref{lem:BSLDAchievable} and the one optimized numerically for $L=500$ layers. For $L=2$, we observe a tight match between $\hat{\Delta}^{2}_{bs}(b,\nu)$ and $\Delta^{2}_{bs}(b,\nu)$ in the high bandwidth ratio regime. However, for intermediate bandwidth ratios, $\hat{\Delta}^{2}_{bs}(b,\nu)$ is significantly worse than $\Delta^{2}_{bs}(b,\nu)$, and, in general, worse than $\Delta_{ld}(b,\nu)$.  Note that, as expected, if the power allocation and the multiplexing gains are jointly optimized, using two layers provides an improvement on the distortion exponent, i.e., $\Delta_{bs}^2(b,\nu)$ outperforms $\Delta_{ld}(b,\nu)$. We also observe that $\hat{\Delta}_{bs}^2(b,\nu)$ and $\hat{\Delta}_{bs}(b,\nu)$ are discontinuous at $b=2.5$, while this discontinuity is not present in the numerically optimized distortion exponents.

Our extensive numerical simulations suggest that, for $b$ values satisfying (\ref{eq:CondBDreduced}), the performance of $\Delta^{L}_{bs}(b,\nu)$ reduces to the distortion exponent achievable by a single layer. We also observe that as the number of layers increases, the difference between $\Delta_{bs}^{L}(b,\nu)$ and $\hat{\Delta}_{bs}^{L}(b,\nu)$ is reduced, and that the distortion exponent achievable by BS-LD as stated in Lemma \ref{lem:BSLDAchievable}, i.e., $\bar{\Delta}_{bs}(b,\nu)$, is indeed very close to the optimal performance that can be achieved by jointly optimizing the multiplexing gain and the power allocation. In the next section, we will see that in certain cases fixing the multiplexing gain to $\hat{\boldsymbol{r}}$ suffices for BS-LD to meet the partially informed upper bound in the MISO/SIMO setup, and therefore $\Delta_{bs}^*(b,\nu)=\hat{\Delta}_{bs}^{\infty}(b,\nu)$.

\begin{figure}
\centering
\includegraphics[width=0.65\textwidth]{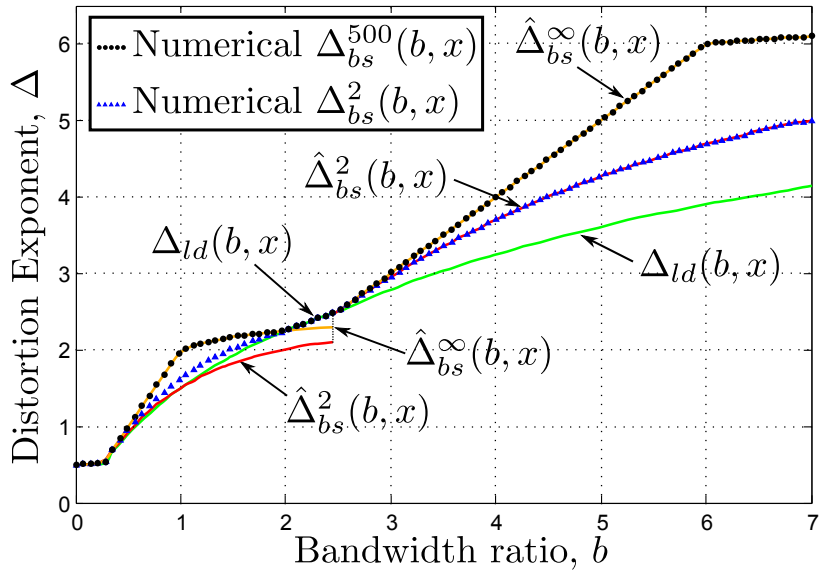}
\vspace{-3 mm}\caption{Distortion exponent achieved by BS-LD with $L=1,2$ and in the limit of infinite layers with respect to
the bandwidth ratio $b$ for a $3\times 2$ MIMO system and a side
information quality given by $\nu=0.5$. Numerical results on the achievable distortion exponent for $L=2$ and $L=500$ are also included.} \vspace{-3 mm}\label{fig:MIMODistortionExponentMultiBad}
\end{figure}

\section{Comparison of the Proposed Schemes and Discussion}\label{sec:Comments}
In this section, we compare the performances of the proposed schemes with each other and with the proposed upper bound. First, we use the upper bound derived in Section \ref{sec:UpperBounds} to characterize the optimal distortion exponent for bandwidth ratios $0\leq b\leq\max\{M^*-M_*+1,\nu\}/M_* $. We show that, when $0\leq b\leq \nu/M_*$, the optimal distortion exponent is achieved by ignoring the channel, and reconstructing the source sequence using only the side information. If  $\nu/M_*\leq b\leq (M^*-M_*+1)/M_*$, then the optimal distortion exponent is achieved by ignoring the side information, and employing the optimal transmission scheme in the absence of side information.

 Then, we characterize the optimal distortion exponent for MISO/SIMO/SISO scenarios.  In MISO/SIMO, i.e., $M_*=1$, we show that BS-LD meets the upper bound, thus characterizing the optimal distortion exponent. This extends the result of \cite{gunduz2008joint} to the case with time-varying source side information. For SISO, i.e., $M^*=M_*=1$, HDA-LD also achieves the optimal distortion exponent. For the general MIMO setup, none of the proposed schemes meet the upper bound for $b>1/M_*$. Nevertheless, multi-layer transmission schemes perform close to the upper bound, especially in the high bandwidth ratio regime.

\subsection{Optimal distortion exponent for low bandwidth ratios}\label{sec:OptimalLowB}\label{sec:LowBand}

First, we consider the MMSE reconstruction of $S^m$ only from the side information sequence $T^m$ available at the receiver, i.e., $\hat{S}_i=\mathrm{E}[S_i|T_i]$. The source sequence is reconstructed with  distortion $D_{no}(\gamma)\triangleq(1+\rho_s\gamma)^{-1}$, and averaging over the side information realizations the distortion exponent is found as $\Delta_{no}(b,\nu)=\nu$, which meets  the upper bound $\Delta_{up}(b,\nu)$ for $0\leq b\leq \nu/M_*$, characterizing the optimal distortion exponent.

\begin{lemma}\label{rm:RemarkOptDist0}
For $0\leq\! b\!\leq \nu/M_*$, the optimal distortion exponent $\Delta^*(b,\nu)\!=\!\nu$ is achieved by simple MMSE reconstruction of $S^m$  from the side information sequence $T^m$.
\end{lemma}

Additionally, Theorem \ref{the:UpperBound} reveals that in certain regimes, the distortion exponent is upper bounded by $\Delta_{\mathrm{MIMO}}(b)$, the distortion exponent upper bound in the absence of side information at the destination \cite[Theorem 3.1]{gunduz2008joint}. In fact, for $\nu/M_*\leq \!b\!\leq M^*-M_*+1$, \mbox{we have $\Delta_{up}(b,\nu)\!=\!bM_*$}. This distortion exponent is achievable for $b$ satisfying $\nu/M_*\leq b\leq (M^*-M_*+1)/M_*$ by ignoring the side information and employing  the optimal scheme in the absence of side information, which is given by the multi-layer broadcast transmission scheme considered in \cite{bhattad2008distortion}. The same distortion exponent is achievable by considering BS-LD ignoring the side information, i.e., $\Delta^*(b,\nu)=\hat{\Delta}^{L}_{bs}(b,0)$. If $\nu/M_*\leq b \leq 1/M_*$, the optimal distortion exponent is also achievable by HDA-S and $\Delta^*(b,\nu)=\Delta_{h}(b,\nu)$.

\begin{lemma}\label{rm:RemarkOptDist}
For $\nu/M_*\leq b \leq (M^*-M_*+1)/M_*$, the optimal distortion exponent is given by  $\Delta^*(b,\nu)=bM_*$, and is achievable by BS-LD ignoring the side information sequence $T^m$. If $\nu/M_*\leq b \leq 1/M_*$ the distortion exponent is also achievable by HDA-S.
\end{lemma}

\subsection{Optimal distortion exponent for MISO/SIMO/SISO}
%
%
%

\begin{figure}
\centering
\includegraphics[width=0.65\textwidth]{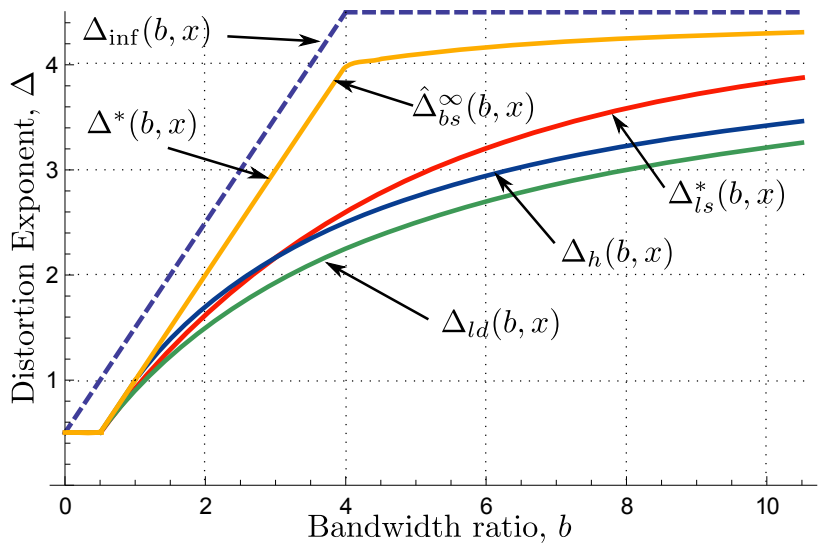}
\vspace{-3 mm}\caption{Distortion exponent $\Delta(b,\nu)$ with respect to
the bandwidth ratio $b$ for a $4\times 1$ MISO system and a side
information quality given by $\nu=0.5$.} \vspace{-3 mm}\label{fig:MISODistortionExponent}
\end{figure}

The following distortion exponent is achievable by HDA-S for $\nu\leq b\leq1$, and by HDA-LD for $b>1$, in the MISO/SIMO setup.
\begin{IEEEeqnarray}{rCl}
\Delta_h(b,\nu)=\begin{cases}\max\{\nu,b\}&\text{for }b\leq 1,\\
\max\{\nu,\frac{M^*+(b-1)(M^*+\nu)}{M^*+b-1}\}&\text{for }b>1.
\end{cases}
\end{IEEEeqnarray}
As seen in Section \ref{sec:OptimalLowB}, HDA-S meets the partially informed upper bound for $b\leq 1$. HDA-LD is in general suboptimal.

 For the multi-layer transmission schemes, the distortion exponent acheivable by  LS-LD is given by
\begin{IEEEeqnarray}{rCl}
\Delta^*_{ls}(b,\nu)=\nu\!+\!M^*\left(1-e^{-\frac{b(1-\kappa^*)}{M^*}}\right),\qquad \kappa^*\!=\!\frac{M^*}{b}\mathcal{W}\left(\frac{e^{\frac{b}{M^*}\nu}}{M^*}\right). 
\end{IEEEeqnarray}

As for BS-LD, considering the achievable rate in Corollary \ref{cor:MlInfLay}, this scheme meets the partially informed encoder lower bound in the limit of infinite layers, i.e., $\hat{\Delta}_{bs}^{\infty}(b,\nu)=\Delta^*_{up}(b,\nu)$. This fully characterizes the optimal distortion exponent in the MISO/SIMO setup, as stated in the next theorem.


\begin{theorem}\label{th:OptDistMISO} The optimal distortion exponent $\Delta^*(b,\nu)$ for MISO/SIMO systems is given by
\begin{IEEEeqnarray}{rCl}
\Delta^*(b,\nu)=\begin{cases}\max\{\nu,b\}&\text{for }b\leq\max\{M^*,\nu\},\\
M^*+\nu\left(1-\frac{M^*}{b}\right)&\text{for }b>\max\{M^*,\nu\},
\end{cases}
\end{IEEEeqnarray}
and is achieved by BS-LD in the limit of infinite layers.
\end{theorem}

We note that in SISO setups, HDA-LD achieves the optimal distortion exponent for $b\geq1$, in addition to BS-LD.
\begin{lemma}
The optimal distortion exponent for SISO channels is achieved by BS-LD, HDA-LD and HDA-S.
\end{lemma}

In Figure \ref{fig:MISODistortionExponent} we plot the distortion exponent for a MISO/SIMO channel with $M^*=4$ and $\nu=0.5$, with respect to the bandwidth ratio $b$. We observe that, as given in Theorem \ref{th:OptDistMISO}, BS-LD achieves the optimal distortion exponent. 
We observe that HDA-LD outperforms LD in all regimes, and, although it outperforms the multi-layer LS-LD for low $b$ values, LS-LD achieves higher distortion exponents than HDA-LD for $b\geq3$.
In general we observe that single-layer schemes perform poorly as the bandwidth ratio increases, as they are not capable of fully exploiting the available degrees-of-freedom in the system.

\begin{figure}
\centering
\includegraphics[width=0.65\textwidth]{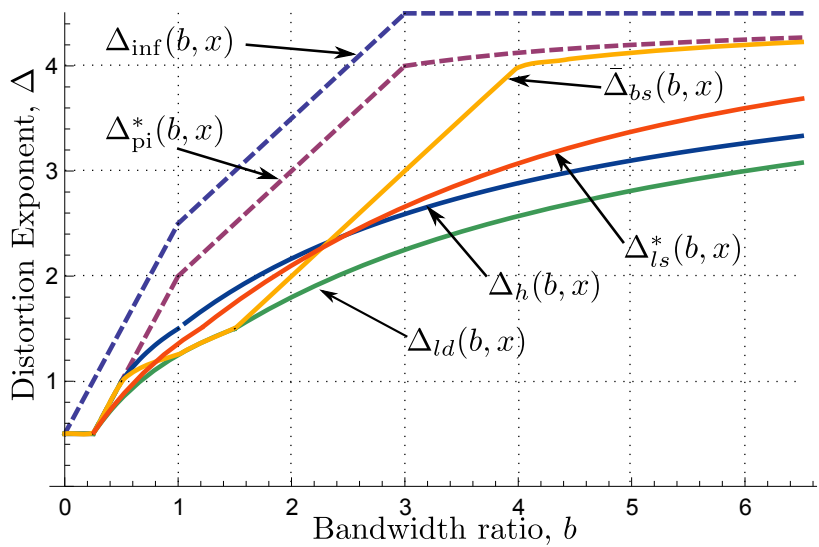}
\vspace{-3 mm}\caption{Distortion exponent $\Delta(b,\nu)$ with respect to
the bandwidth ratio $b$ for a $2\times 2$ MIMO system and a side
information quality given by $\nu=0.5$.} \vspace{-3 mm}\label{fig:MIMODistortionExponent}
\end{figure}

%


\subsection{General MIMO}

Here, we consider the general MIMO channel with $M_*>1$. Figure \ref{fig:MIMODistortionExponent} shows the upper and lower bounds on the distortion exponent derived in the previous sections for a $2\times 2$ MIMO channel with $\nu=0.5$. First, it is observed that the optimal distortion exponent is achieved by HDA-S and BS-LD with infinite layers for $b\leq 0.5$, as expected from Section \ref{sec:OptimalLowB}, while the other schemes are suboptimal in general.

\begin{figure}
\centering
\includegraphics[width=0.65\textwidth]{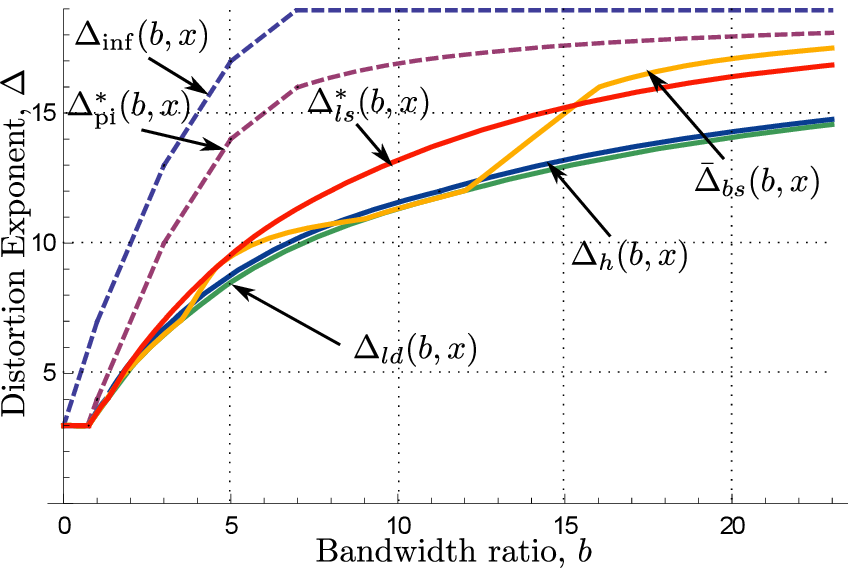}
\vspace{-3 mm}\caption{Distortion exponent $\Delta(b,\nu)$ with respect to
the bandwidth ratio $b$ for a $4\times 4$ MIMO system and a side
information quality given by $\nu=3$.} \vspace{-3 mm}\label{fig:MIMODistortionExponent2}
\end{figure}

For $0.5<b\lesssim 2.4$, HDA-LD is the scheme achieving the highest distortion exponent, and outperforms BS-LD, and in particular, when the performance of BS-LD reduces to that of LD, since HDA-LD outperforms LD in general. For larger $b$ values, the highest distortion exponent is achieved by BS-LD. Note that for $b\geq4$, $\Delta_{bs}^*(b,0.5)$ is very close to the upper bound. We also observe that for $b\gtrsim 2.4$ LS-LD outperforms HDA-LD, but it is worse than BS-LD. This is not always the case, as will be seen next.

\begin{figure}
\centering
\includegraphics[width=0.65\textwidth]{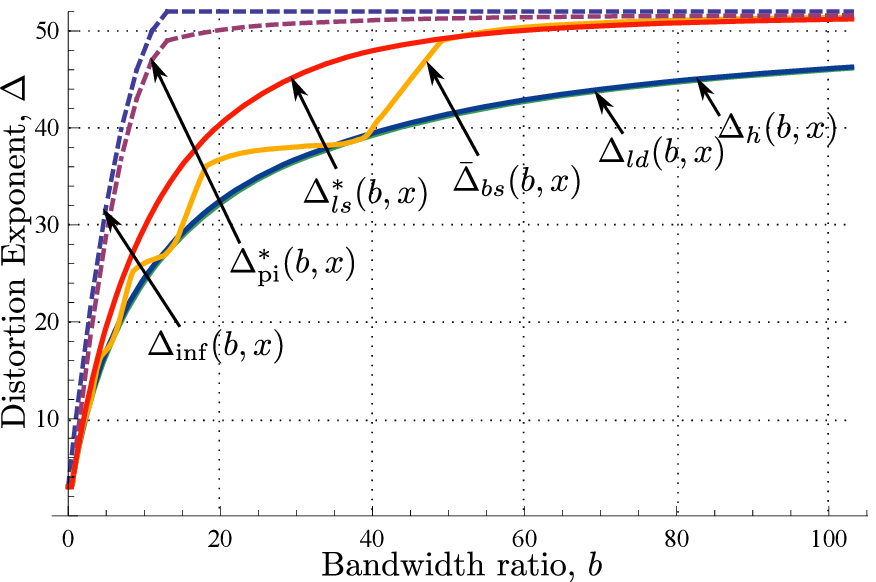}
\vspace{-3 mm}\caption{Distortion exponent $\Delta(b,\nu)$ with respect to
the bandwidth ratio $b$ for a $4\times 4$ MIMO system and a side
information quality given by $\nu=3$.} \vspace{-3 mm}\label{fig:MIMODistortionExponentManyAnt}
\end{figure}
In Figure \ref{fig:MIMODistortionExponent2}, we plot the upper and lower bounds for a $4\times 4$ MIMO channel with $\nu=3$. We note that, for $b\leq \max\{1,\nu\}/M_*$, $\Delta^*(b,3)=3$, which is achievable by using only the side information sequence at the decoder. For this setup, LS-LD achieves the best distortion exponent for intermediate $b$ values, outperforming both HDA-LD and BS-LD. Again, in the large bandwidth ratio regime, BS-LD achieves the best distortion exponent values, and performs close to the upper bound. We note that for high side information quality, the difference in performance between LD and HDA-LD decreases. 
Comparing Figure \ref{fig:MIMODistortionExponent} and Figure \ref{fig:MIMODistortionExponent2}, we observe that, when the side information quality is high, digital schemes better exploit the degrees-of-freedom of the system than analog schemes.

In Figure \ref{fig:MIMODistortionExponentManyAnt} we plot the upper and lower bounds for a $7\times 7$ MIMO channel with $\nu=3$.  In comparison with Figure \ref{fig:MIMODistortionExponent2}, as the number of antennas increases the difference in performance between LD and HDA-LD decreases. This seems to be the case also between BS-LD and LS-LD in the high bandwidth ratio regime. However, LS-LD significantly outperforms LS-LD for intermediate $b$ values. 

\section{Conclusions}\label{sec:Conclusions}
We have studied the high SNR distortion exponent when transmitting a Gaussian source over a time-varying MIMO fading channel in the presence of time-varying correlated side information at the receiver. We consider a block-fading model for both the channel and the source side information states, and assume that perfect state information is available at the receiver, while the transmitter has only a statistical knowledge. We have derived an upper bound on the distortion exponent, as well as lower bounds based on separate source and channel coding, list decoding and hybrid digital-analog transmission. We have also proposed multi-layer transmission schemes based on progressive transmission with joint decoding as well as  superposition with list decoding. We have considered the effects of the bandwidth ratio and the side information quality on the distortion exponent, and shown that the multi-layer superposition transmission meets the upper bound in MISO/SIMO/SISO channels, solving the joint source channel coding problem in the high SNR regime. For general MIMO channels, we have characterized the optimal distortion exponent in the low bandwidth ratio regime, and shown that the multi-layer scheme with superposition performs very close to the upper bound in the high bandwidth ratio regime.

\begin{appendices}


\section{Proof of Theorem  \ref{the:UpperBound}} \label{app:UpperBound}
The exponential integral can be bounded as follows \cite[p.229, 5.1.20]{abramowitz2012handbook}:
\begin{IEEEeqnarray}{rCl}
\frac{1}{2}\ln\left(1+\frac{2}{t}\right)<e^{t}E_{1}(t)<\ln\left(1+\frac{1}{t}\right),
\quad t>0.
\end{IEEEeqnarray}

Next, using the lower bound $\ln (1+t)\geq \frac{t}{1+t}$, for $t> -1$, we have
\begin{IEEEeqnarray}{rCl}
\frac{1}{2}\ln\left(1+\frac{2}{t}\right)>\frac{1}{2}\frac{2/t}{1+2/t}=\frac{1}{t+2}.
\end{IEEEeqnarray}

Then, $ED^*_{\text{pi}}$ in (\ref{eq:LowerBound}) is lower bounded by
\begin{IEEEeqnarray}{rCl}\label{eq:Bound}
ED^*_{\text{pi}}(\rho,\rho_s,b)\geq\int_{\mathbf{H}}\frac{1}{2^{b\mathcal{C}(\mathbf{H})}+2\rho_s}p_{h}(\mathbf{H})d\mathbf{H}.
\end{IEEEeqnarray}

Following \cite{zheng2003diversity}, the capacity of the MIMO channel is upper bounded as
\begin{IEEEeqnarray}{rCl}\label{eq:CapBound}
\mathcal{C}(\mathbf{H})&=&\sup_{\mathbf{C}_u:\text{Tr}\{\mathbf{C}_u\}\leq M_t}\log\det\left(\mathbf{I}+\frac{\rho}{M_t}\mathbf{HC}_u\mathbf{H}^H\right) \\
&\leq& \log\det \left(\mathbf{I}+\rho \mathbf{HH}^{H}\right),
\end{IEEEeqnarray}
where the inequality follows from the fact that $M_t\mathbf{I}-\mathbf{C}_u\succeq 0$ subject to the power constraint $\text{Tr}\{\mathbf{C}_u\}\leq M_t$, and the function $\log \det (\cdot)$ is nondecreasing on the cone of positive semidefinite Hermitian matrices.

Let $\lambda_{M_{*}}\geq \cdots \geq\lambda_{1}>0$ be the eigenvalues of matrix $\mathbf{HH}^{H}$, and consider the change of variables $\lambda_{i}=\rho^{-\alpha_{i}}$, with $\alpha_1\geq...\geq
\alpha_{M_*}\geq 0$. The joint probability density function (pdf) of
$\boldsymbol\alpha\triangleq[\alpha_{1},...,\alpha_{M_{*}}]$ is given by \cite{zheng2003diversity}:
\begin{IEEEeqnarray}{rCl}\label{eq:alphapdf}
p_A(\boldsymbol\alpha)&=&K^{-1}_{M_t,M_r}(\log\rho)^{M_{*}}\prod_{i=1}^{M_{*}}\rho^{-(M^{*}-M_{*}+1)\alpha_{i}}
\cdot
\left[\prod_{i<j}(\rho^{\alpha_{i}}-\rho^{\alpha_{j}})^2\right]\exp\left(-\sum_{i=1}^{M_{*}}\rho^{\alpha_{i}}\right),
\end{IEEEeqnarray}
where $K^{-1}_{M_t,M_r}$ is a normalizing constant.

We define the high SNR exponent of $p_{A}(\boldsymbol\alpha)$ as $S_{A}(\boldsymbol\alpha)$, that is, we have $p_{A}(\boldsymbol\alpha)\doteq\rho^{-S_{A}(\boldsymbol\alpha)}$, where
\begin{IEEEeqnarray}{rCl}
S_A(\boldsymbol\alpha)\!\triangleq\!\begin{cases}
\sum_{i=1}^{M_{*}}(2i-1+M^{*}-M_{*})\alpha_{i}&\text{if }\alpha_{M_*}\!\geq\!0, \\
\infty&\text{otherwise}.
\end{cases}
\end{IEEEeqnarray}

Then, from (\ref{eq:Bound}) and (\ref{eq:CapBound}) we have
\begin{IEEEeqnarray}{rCl}
ED^*_{\text{pi}}(\rho,\rho_s,b)&\geq&\int_{\mathbf{H}}\frac{1}{\prod_{i=1}^{M_*}\left(1+\rho \lambda_i\right)^b+2\rho_s}p_{h}(\mathbf{H})d\mathbf{H}\\
&=&\int_{\mathbf{\boldsymbol\alpha}}\frac{1}{\prod_{i=1}^{M_*}\left(1+\rho^{1-\alpha_i}\right)^{b}+2\rho_{s}}p_A(\boldsymbol\alpha)d\boldsymbol\alpha\\
&\geq&\int_{\mathbf{\boldsymbol\alpha}^+}G_{\rho}(\boldsymbol\alpha)p_A(\boldsymbol\alpha)d\boldsymbol\alpha\label{eq:EqBound1},
\end{IEEEeqnarray}
where we define
\begin{IEEEeqnarray}{rCl}
G_{\rho}(\boldsymbol\alpha)\triangleq \left(\prod_{i=1}^{M_*}\left(1+\rho^{1-\alpha_i}\right)^{b}+2\rho_{s}\right)^{-1},
\end{IEEEeqnarray}
and the set
$\boldsymbol\alpha^+\triangleq\{\boldsymbol\alpha\in \mathds{R}^{M_*}:1 \geq \alpha_1\geq\cdots\geq
\alpha_{M_*}\geq 0\}$ in (\ref{eq:EqBound1}).

Then, the distortion exponent of the partially informed encoder is upper bounded by
\begin{IEEEeqnarray}{rCl}
\Delta_{pi}^*(b,\nu) &\triangleq&-\lim_{\rho\rightarrow\infty}\frac{\log ED^*_{\text{pi}}(\rho,\rho_s,b)}{\log\rho}\\
&\leq& \lim_{\rho\rightarrow\infty}\frac{1}{\log\rho}\log\int_{\boldsymbol\alpha^+}\text{exp}\left(\frac{\log G_{\rho}(\boldsymbol\alpha)}{\log\rho}\log\rho \right)p_A(\boldsymbol\alpha)d\boldsymbol\alpha\\
&=& \lim_{\rho\rightarrow\infty}\frac{1}{\log\rho}\log\int_{\mathbf{\boldsymbol\alpha}^+}\text{exp}\left(G(\boldsymbol\alpha)\log \rho \right)p_A(\boldsymbol\alpha)d\boldsymbol\alpha,\label{eq:EqCDT}
\end{IEEEeqnarray}
where (\ref{eq:EqCDT}) follows from the application of the Dominated Convergence Theorem \cite{Durrett:book}, which holds since $G_{\rho}(\boldsymbol\alpha)\leq 1$ for all $\boldsymbol\alpha$, the continuity of the logarithmic and exponential functions, and since we have the following limit
\begin{IEEEeqnarray}{rCl}
G(\boldsymbol\alpha)&\triangleq&\lim_{\rho\rightarrow\infty}\frac{\log G_{\rho}(\boldsymbol\alpha)}{\log \rho}
=\lim_{\rho\rightarrow\infty}\frac{\log(\rho^{b\sum_{i=1}^{M_*}(1-\alpha_i)^+}+2\rho^{\nu})^{-1}}{\log \rho}\\
&=&\begin{cases}
-\nu&\text{if }\nu>b\sum_{i=1}^{M_*}(1-\alpha_i)^+,\\
-b\sum_{i=1}^{M_*}(1-\alpha_i)^+&\text{if }\nu\leq b\sum_{i=1}^{M_*}(1-\alpha_i)^+,
\end{cases}
\end{IEEEeqnarray}
where we have used the exponential equalities $1+\rho^{1-\alpha_i}\doteq \rho^{(1-\alpha_i)^+}$, and $\rho_s\doteq\rho^\nu$.


From Varadhan's lemma \cite{Dembo:book}, it follows that the distortion exponent of $ED^*_{\text{pi}}$ is upper bounded by the solution to the following optimization problem,
\begin{IEEEeqnarray}{rCl}\label{eq:ExponentBoundProblem}
\Delta_{up}(b,\nu)\triangleq \inf_{\boldsymbol\alpha^+}[-G(\boldsymbol\alpha)+S_A(\boldsymbol\alpha)].
\end{IEEEeqnarray}
In order to solve (\ref{eq:ExponentBoundProblem}) we divide the optimization into two subproblems: the case when $\nu< b\sum_{i=1}^{M_*}(1-\alpha_i)$, and the case when $\nu\geq b\sum_{i=1}^{M_*}(1-\alpha_i)$. The solution is then given by the minimum of the solutions of these subproblems.

If $\nu\geq b\sum_{i=1}^{M_*}(1-\alpha_i)$, the problem in (\ref{eq:ExponentBoundProblem}) reduces to
\begin{IEEEeqnarray}{rCl}\label{eq:DeltaBoundProblem2}
\Delta^1_{up}(b,\nu)= \nu +&&\inf_{\boldsymbol\alpha^+}\sum_{i=1}^{M_{*}}(2i-1+M^{*}-M_{*})\alpha_{i}
  \\
&&\text{s.t. }\sum_{i=1}^{M_*}(1-\alpha_i)\leq\frac{ \nu}{b}.
\end{IEEEeqnarray}
The optimization in (\ref{eq:DeltaBoundProblem2}) can be identified with the DMT problem in \cite[Eq. (14) ]{zheng2003diversity} for a multiplexing gain of $r=\frac{ \nu}{b}$. Next, we give an explicit solution for completeness.

First, if $bM_*\leq  \nu$, the infimum is given by $\Delta^1_{up}(b,\nu)= \nu$
for $\boldsymbol\alpha^*=0$. Then, for $k\leq \frac{ \nu}{b}\leq k+1$, for $k=0,...,M_*-1$, i.e.,
$\frac{ \nu}{k+1}\leq b\leq \frac{ \nu}{k}$, the infimum is achieved by
\begin{IEEEeqnarray}{rCl}
\alpha^*_i=\begin{cases} 1 &\text{for }i=1,...,M_*-k-1,\\
 k+1-\frac{ \nu}{b}&i=M_*-k,\\
0&\text{for }i=M_*-k+1,...,M_*.
\end{cases}
\end{IEEEeqnarray}
Substituting, we have, for $k=0,...,M_*-1$,
\begin{IEEEeqnarray}{rCl}
\Delta^1_{up}(b,\nu)&=& \nu +\Phi_k-\Upsilon_k\left(\frac{ \nu}{b}-k\right)= \nu +d^*\left(\frac{ \nu}{b}\right),
\end{IEEEeqnarray}
where $\Phi_k$ and $\Upsilon_k$ are defined as in (\ref{eq:ML_PhiUpsilon}).

Now we solve the second subproblem with $\nu < b\sum_{i=1}^{M_*}(1-\alpha_i)$.  Since $1 \geq \alpha_1\geq...\geq
\alpha_{M_*}\geq 0$ we can  rewrite (\ref{eq:ExponentBoundProblem}) as
\begin{IEEEeqnarray}{rCl}\label{eq:BoundDist1}
\Delta^2_{up}(b,\nu)&=&\inf_{\boldsymbol\alpha^+} bM_* -\sum_{i=1}^{M_*}\alpha_i\phi(i)  \\
&&\text{s.t. }\sum_{i=1}^{M_*}\alpha_i<M_*-\frac{ \nu}{b},
\end{IEEEeqnarray}
where we have defined $\phi(i)\triangleq [b-(2i-1+M^*-M_*)]$. Note that $\phi(1)>\cdots>\phi(M_*)$.

First, we note that for $bM_*< \nu$ there is no feasible solution due to the constraint in (\ref{eq:BoundDist1}).

Now, we consider the case $\nu \leq M_*(1+M^*-M_*)$.
 If $\frac{ \nu}{M_*}\leq b< 1+M^*-M_*$, all the terms $\phi(i)$ multiplying $\alpha_i$'s are negative, and, thus, the infimum is achieved by $\boldsymbol\alpha^*=0$, and is given by $\Delta^2_{up}(b,\nu)=b M_* $.
If $1+M^*-M_*\leq b <3+M^*-M_*$, then $\phi(1)$ multiplying $\alpha_1$ is positive, while the other $\phi(i)$ terms are negative. Then $\alpha_i^*=0$ for $i=2,...,M_*$. From (\ref{eq:BoundDist1}) we have $\alpha_1\leq M_*-\frac{ \nu}{b}$. If $b\geq \frac{ \nu}{M_*-1}$, the right hand side (r.h.s.) of (\ref{eq:BoundDist1}) is greater than one, and smaller otherwise. Then, we have
\begin{IEEEeqnarray}{rCl}
\alpha_1^*=\begin{cases}
1&\text{if }b\geq \frac{ \nu}{M_*-1},\\
M_*-\frac{ \nu}{b}&\text{if }b< \frac{ \nu}{M_*-1}.
\end{cases}
\end{IEEEeqnarray}
Note that $\alpha^*_1\geq0$ since $b>\frac{ \nu}{M_*}$.

When $2k-1+M^*-M_*\leq b <2k+1+M^*-M_*$ for $k=2,..., M_*-1$, the coefficients $\phi(i)$, $i=1,...,k$, associated with the first $k$ $\alpha_i$ terms are positive, while the others remain negative. Then, \begin{IEEEeqnarray}{rCl}\label{eq:alphaopt}
\alpha_i^*=0, \quad \text{for } i=k+1,...,M_*.
\end{IEEEeqnarray}

Since $\phi(i)$, $i=1,...,k$, are positive and $\phi(1)>\cdots>\phi(k)$, we have
$\alpha^*_i=1$ for $i=1,...,k-1$, and the constraint becomes $\alpha_k<M_*-(k-1)-\frac{ \nu}{b}$. If $b\geq \frac{ \nu}{M_*-k}$, then the r.h.s. is greater than one, and smaller otherwise. In order for the solution to be feasible, we need  $\alpha_k\geq0$, that is, $M_*-(k-1)-\frac{ \nu}{b}\geq0$. Then we have
\begin{IEEEeqnarray}{rCl}\label{eq:SolOne}
\alpha^*_k=\begin{cases}1&\text{if }b\geq \frac{ \nu}{M_*-k},\\
M_*-(k-1)-\frac{ \nu}{b}&\text{if }\frac{ \nu}{M_*-(k-1)}\leq b<\frac{ \nu}{M_*-k}.
\end{cases}
\end{IEEEeqnarray}

If $b<\frac{ \nu}{M_*-(k-1)}$, the solution in (\ref{eq:SolOne}) is not feasible. Instead, we have $\alpha^*_k=0$, since $\phi(k)<\phi(k-1)$, $\alpha^*_i=0$ for $i=k+1,...,M_*$, and $\alpha^*_i=1$, for $i=1,...,k-2$. Then, the constraint in (\ref{eq:BoundDist1}) is given by $\alpha_{k-1}\leq M_*-(k-2)-\frac{ \nu}{b}$. Since $b<\frac{ \nu}{M_*-(k-1)}$, the r.h.s. is always smaller than one. For the existence of a feasible solution, the r.h.s. is required to be greater than zero. Therefore, we have
\begin{IEEEeqnarray}{rCl}
\alpha^*_{k-1}&=&M_*-(k-2)-\frac{ \nu}{b}, \qquad
\text{if }\frac{ \nu}{M_*-(k-2)}\leq b<\frac{ \nu}{M_*-(k-1)}.
\end{IEEEeqnarray}
In general, iterating this procedure, for
\begin{IEEEeqnarray}{rCl}
\frac{ \nu}{M_*-(j-1)}\leq b< \frac{ \nu}{M_*-j},\quad j=1,...,k,
\end{IEEEeqnarray}
we have
\begin{IEEEeqnarray}{rCl}\label{eq:DMT2Sol}
\alpha^*_i=\begin{cases}
1& \text{for }i=1,...,j-1,\\
M_*-(j-1)-\frac{ \nu}{b}&\text{for }i=j,\\
0&\text{for }i=j+1,...,M_*.
\end{cases}
\end{IEEEeqnarray}

Note that for the case $j=1$, we have $\alpha_1=M_*-\frac{ \nu}{b}$, which is always feasible.

 We now evaluate (\ref{eq:BoundDist1}) with the optimal $\boldsymbol\alpha^*$  if $2k-1+M^*-M_*\leq b <2k+1+M^*-M_*$ for some  $k\in\{2,..., M_*-1\}$. For $b\geq \frac{ \nu}{M_*-k}$, we have $\alpha_1=\cdots=\alpha_k=1$ and $\alpha_{k+1}=\cdots=\alpha_{M_*}=0$, and then
\begin{IEEEeqnarray}{rCl}
\Delta^2_{up}(b,\nu)&=&\sum_{i=1}^{M_*}\min\{b,2i-1+M^*-M_*\}=\Delta_{\text{MIMO}}(b).
\end{IEEEeqnarray}

For $\frac{ \nu}{M_*}\leq b\leq \frac{ \nu}{M_*-k}$, substituting (\ref{eq:DMT2Sol}) into  (\ref{eq:BoundDist1}) we have
\begin{IEEEeqnarray}{rCl}
\Delta^2_{up}(b,\nu)&=& \nu +(M^*-M_*-1+j)(j-1)\\
&&+\left(M_*-(j-1)-\frac{ \nu}{b}\right)(2j-1+M^*-M_*) ,
\end{IEEEeqnarray}
where
\begin{IEEEeqnarray}{rCl}
\frac{ \nu}{M_*-(j-1)}\leq b\leq \frac{ \nu}{M_*-j},\quad \text{for some }j\in\{1,...,k\}.
\end{IEEEeqnarray}

Note that with the change of index $j=M_*-j'$, we have, after some manipulation,
\begin{IEEEeqnarray}{rCl}
\Delta^2_{up}(b,\nu)&=& \nu +(M^*-j')(M_*-j')
-\left(\frac{ \nu}{b}-j'\right)(M^*+M_*-2j'-1),
\end{IEEEeqnarray}
in the regime
\begin{IEEEeqnarray}{rCl}
\frac{ \nu}{j'+1}\leq b< \frac{ \nu}{j'},\quad j'=M_*-k,...,M_*-1.
\end{IEEEeqnarray}

This is equivalent to the value of the DMT curve in (\ref{eq:DMTcurve}) at multiplexing gain $r=\frac{ \nu}{b}$. Then, for $\frac{ \nu}{M_*}\leq b< \frac{ \nu}{M_*-k}$ we have
\begin{IEEEeqnarray}{rCl}
\Delta^2_{up}(b,\nu)&=& \nu +d^{*}\left(\frac{ \nu}{b}\right).
\end{IEEEeqnarray}

If $b\geq M^*+M_*-1$, the infimum is achieved by $\alpha_i^*=1$, for $i=1,...,M_*-1$, and $\alpha^*_{M_*}=1-\frac{ \nu}{b}$ if $b\geq  \nu$. If $b <  \nu$, this solution is not feasible, and the solution is given by (\ref{eq:DMT2Sol}). Therefore, in this regime we also have
\begin{IEEEeqnarray}{rCl}
\Delta^2_{up}(b,\nu)&=& \nu +d^{*}\left(\frac{ \nu}{b}\right).
\end{IEEEeqnarray}

Putting all these results together,  for $\nu \leq M_*(M^*-M_*+1)$ we have
\begin{IEEEeqnarray}{lCl}
\Delta^2_{up}(b,\nu)=\begin{cases}
bM_*&\text{for }\frac{ \nu}{M_*}\leq b< M^*-M_*+1,\\
 \nu +d^*\left(\frac{ \nu}{b}\right)&\text{for
}M^*-M_*+1\leq b<\frac{ \nu}{M_*-k},\\
\Delta_{\text{MIMO}}(b)&\text{for
}\frac{ \nu}{M_*-k}\leq b<M^*+M_*-1,\\
 \nu +d^{*}\left(\frac{ \nu}{b}\right)&\text{for
}b\geq M^*+M_*-1,
\end{cases}
\end{IEEEeqnarray}
where $k\in\{1,...,M_*-1\}$  is the integer satisfying $2k-1+M^*-M_*\leq b <2k+1+M^*-M_*$.

 Now, we solve (\ref{eq:BoundDist1}) for  $M_*(M^*-M_*+1)\leq x< M_*(M^*+M_*-1)$. Let $l\in\{2,...,M_*\}$ be the integer satisfying $M_*(2(l-1)-1+M^*-M_*)\leq x< M_*(2l-1+M^*-M_*)$. The first interval of
$b$ in which a feasible solution exists is given by $\frac{ \nu}{M_*}\leq b <
2l-1+M^*-M_*$. From the sign of the coefficients $\phi(i)$ in this interval we have $\alpha^*_{i}=0$ for
$i=(l+1),...,M_*$, and $\alpha^*_i=1$ for $i=1,..., l-1$. Substituting, the
constraint becomes $\alpha_l<M_*-(l-1)-\frac{ \nu}{b}$.
If $b>\frac{ \nu}{M_*-l}$ the r.h.s. is larger than one, and
$\alpha^*_l=1$. On the contrary, if $b\leq\frac{ \nu}{M_*-l}$, it is given
by $\alpha^*_{l}=M_*-(l-1)-\frac{ \nu}{b}$ if $b>\frac{ \nu}{M_*-(l-1)}$, so that the r.h.s. of the constraint is larger than zero. Iterating this procedure, the solution for all $b$ values is found following the techniques that lead to (\ref{eq:DMT2Sol}).
In general, for  $2k-1+M^*-M_*\leq b <2k+1+M^*-M_*$, $k=l,..., M_*-1$ and
\begin{IEEEeqnarray}{rCl}
\frac{ \nu}{M_*-(j-1)}\leq b< \frac{ \nu}{M_*-j},\quad j=1,...,k,
\end{IEEEeqnarray}
we have
\begin{IEEEeqnarray}{rCl}
\alpha^*_i=\begin{cases}
1& \text{for }i=1,...,j-1,\\
M_*-(j-1)-\frac{ \nu}{b}&\text{for }i=j,\\
0&\text{for }i=j+1,...,M_*.
\end{cases}
\end{IEEEeqnarray}

The distortion exponent is now obtained similarly to the case $\nu \leq M_*(M^*-M_*+1)$ in each interval $2k-1+M^*-M_*\leq b< 2k+1+M^*-M_*$  with  $k=l,...,M_*-1$ instead of $k=1,...,M_*-1$, and thus, we omit the details. Putting all together, if $\nu $ satisfies $M_*(2(l-1)-1+M^*-M_*)\leq x< M_*(2l-1+M^*-M_*)$, for some  $l\in\{2,...,M_*\}$, we have
\begin{IEEEeqnarray}{lCl}
\Delta^2_{up}(x,b)=\begin{cases}
 \nu +d^*\left(\frac{ \nu}{b}\right)&\text{for
}\frac{ \nu}{M_*}\leq b<2l-1+M^*-M_*,\\
 \nu +d^*\left(\frac{ \nu}{b}\right)&\text{for }
2l-1+M^*-M_*\leq b<\frac{ \nu}{M_*-k},\\
\Delta_{\text{MIMO}}(b)&\text{for
}\frac{ \nu}{M_*-k}\leq b<M^*+M_*-1,\\
 \nu +d^{*}\left(\frac{ \nu}{b}\right)&\text{for
}b\geq M^*+M_*-1.
\end{cases}
\end{IEEEeqnarray}
Note that in the case $l=M_*$, we have $\Delta^2_{up}(x,b)= \nu +d^{*}\left(\frac{ \nu}{b}\right)$ for any $b$ value.

Finally, the case $\nu \geq M_*(M^*+M_*-1)$ can be solved similarly. Notice that if $\alpha^*_i=1$, $i=1,...,M_*-1$ we have the constraint $\alpha_{M_*}\leq 1-\frac{ \nu}{b}$, that is, we never have the case $\alpha^*_{M_*}=1$. Then, the optimal $\alpha^*_i$ are given as in (\ref{eq:DMT2Sol}), and we have
\begin{IEEEeqnarray}{lCl}
\Delta^2_{up}(x,b)= \nu +d^*\left(\frac{ \nu}{b}\right)\quad \text{for
}\frac{ \nu}{M_*}\leq b.
\end{IEEEeqnarray}

Now, $\Delta_{up}(b,\nu)$ is given by the minimum of
$\Delta^1_{up}(b,\nu)$ and $\Delta^2_{up}(b,\nu)$. First, we note that $\Delta_{2up}(b,\nu)$  has no feasible solution for $bM_*\leq  \nu$, and we have $\Delta_{up}(b,\nu)=\Delta^1_{up}(b,\nu)= \nu$ in this region. For $bM_*>  \nu$, both solutions $\Delta^1_{up}(b,\nu)$ and $\Delta^2_{up}(b,\nu)$ coincide except in the range $\frac{ \nu}{M_*-k}\leq b\leq M^*+M_*-1$. We note that $\Delta^1_{up}(b,\nu)$ in (\ref{eq:DeltaBoundProblem2}) is linear and increasing in
$\boldsymbol\alpha$, and hence, the solution is such that the
constraint is satisfied with equality, i.e.,
$\nu =\sum_{i=1}^{M_*}b(1-\alpha_i)$. That is, $\Delta^2_{up}(b,\nu)\leq\Delta^1_{up}(b,\nu)$ whenever both solutions exist in the same $\boldsymbol\alpha$ region. Then, the minimizing
$\boldsymbol\alpha$ will be one such that either
$\Delta^1_{up}(b,\nu)<\Delta^2_{up}(b,\nu)$, or the one arbitrarily close to the boundary
 $\nu =b\sum_{i=1}^{M_*}(1-\alpha_i)^+$, where
$\Delta^1_{up}(b,\nu)=\Delta^2_{up}(b,\nu)$. Consequently,
$\min\{\Delta^1_{up}(b,\nu),\Delta^2_{up}(b,\nu)\}=\Delta^1_{up}(b,\nu)$, whenever they are defined in the same region. Putting all the results together we complete the proof.

\section{Proof of Theorem \ref{the:ExponentJoint}}\label{app:ExponentJoint}
Applying the  change
of variables $\lambda_i=\rho^{-\alpha_i}$ and $\gamma=\rho^{-\beta}$, and considering a rate $R_{ld}=r_{ld}\log\rho$, $r_{ld}>0$, the outage event in (\ref{eq:OutageSets}) can be written as
\begin{IEEEeqnarray}{rCl}
\mathcal{O}_{ld}&=&\left\{(\mathbf{H},\gamma): 1+\frac{2^{-\epsilon}\rho^{br_{ld}}-1}{
\gamma\rho^ \nu +
1}\geq\prod_{i=1}^{M_*}(1+\rho\lambda_i)^b \right\}\\
&=&\left\{(\boldsymbol\alpha,\beta):1+\frac{2^{-\epsilon}\rho^{br_{ld}}-1}{\rho^{(\nu-\beta)}+1}\geq\prod_{i=1}^{M_*}(1+\rho^{1-\alpha_i})^b\right\} . 
\end{IEEEeqnarray}

For large $\rho$, we have
\begin{IEEEeqnarray}{lCl}
\frac{1+\frac{2^{-\epsilon}\rho^{br_{ld}}-1}{\rho^{(\nu-\beta)}+1}}{\prod_{i=1}^{M_*}(1+\rho^{1-\alpha_i})^b}
&\doteq&\frac{1+\rho^{br_{ld}}\rho^{-(\nu-\beta)^+}}{\rho^{b\sum_{i=1}^{M_*}(1-\alpha_i)^+}}
\doteq\rho^{({br_{ld}}-(\nu-\beta)^+)^+-b\sum_{i=1}^{M_*}(1-\alpha_i)^+}.
\end{IEEEeqnarray}

Therefore, at high SNR, the  achievable expected end-to-end distortion for LD is found as,
\begin{IEEEeqnarray}{rCl}\label{eq:DIntWin}
ED_{ld}(br_{ld}\log{\rho})&=&\int_{\mathcal{O}_{ld}^c}\!\! D_{d}(br_{ld}\log\rho,\rho^{-\beta}) p_{A}(\boldsymbol\alpha)p_{B}(\beta)d \boldsymbol\alpha d \beta\\
&&+ \int_{\mathcal{O}_{ld}} D_{d}(0,\rho^{-\beta}) p_{A}(\boldsymbol\alpha)p_{B}(\beta) d \boldsymbol\alpha d \beta\\
&\doteq&\int_{\mathcal{A}_j^c} \rho^{-\max\{(\nu-\beta)^+,br_{ld}\}} \rho^{-(S(\boldsymbol\alpha)+\beta)}d \boldsymbol\alpha d \beta \\
&&+ \int_{\mathcal{A}_j} \rho^{-(\nu-\beta)^+} \rho^{-(S(\boldsymbol\alpha)+\beta)} d \boldsymbol\alpha d \beta . \\
&\doteq&\rho^{-\Delta^1_{j}(r_{ld})}+\rho^{-\Delta^2_{j}(r_{ld})} \\
&\doteq&\rho^{-\min\{\Delta^1_{j}(r_{ld}),\Delta^2_{j}(r_{ld})\}} \\
&\doteq&\rho^{-\Delta_{ld}(r_{ld})},
\end{IEEEeqnarray}
where $D_d(R,\gamma)$ is as defined in (\ref{eq:DistSSCC}), and we have used  $D_{d}(r\log\rho,\beta)\doteq\rho^{-\max\{(\nu-\beta)^+,2r\}}$. We have also defined the  high SNR equivalent of the outage event as
\begin{IEEEeqnarray}{rCl}
\mathcal{A}_{j}\triangleq\left\{(\boldsymbol\alpha,\beta):({br_{ld}}-(\nu-\beta)^+)^+\geq b\sum_{i=1}^{M_*}(1-\alpha_i)^+\right\}. 
\end{IEEEeqnarray}
We have applied Varadhan's lemma to each integral to obtain
\begin{IEEEeqnarray}{rCl}\label{eq:Delta1j}
\Delta^1_{j}(r_{ld})\!\triangleq\!\inf_{\mathcal{A}^c_{j}}\max\{(\nu-\beta)^+,br_{ld}\}\!+\beta+S_A(\boldsymbol\alpha),
\end{IEEEeqnarray}
and
\begin{IEEEeqnarray}{rCl}\label{eq:Delta2j}
\Delta^2_{j}(r_{ld})&\triangleq&\inf_{\mathcal{A}_{j}} (\nu-\beta)^++\beta+S_A(\boldsymbol\alpha).
\end{IEEEeqnarray}
Then, the distortion exponent of LD is found as
\begin{align}\label{eq:Delta0j}
\Delta_{ld}(r_{ld})=\min\{\Delta^1_{j}(r_{ld}),\Delta^2_{j}(r_{ld})\}.
\end{align}

 We first solve (\ref{eq:Delta1j}). We can constrain the optimization  to $\boldsymbol\alpha\geq0$ and $\beta\geq0$ without loss of optimality, since for $\boldsymbol\alpha,\beta<0$ we have $S_A(\boldsymbol\alpha)=S_{B}(\beta)=+\infty$.  Then, $\Delta^1_{j}(r_{ld})$ is minimized by  $\boldsymbol\alpha^*=0$ since this minimizes $S_A(\boldsymbol\alpha)$ and enlarges $\mathcal{A}^c_j$. We can rewrite (\ref{eq:Delta1j}) as
  \begin{IEEEeqnarray}{rCl}\label{eq:Delta1jPart}
\Delta^1_{j}(r_{ld})&=&\inf_{\beta\geq0} \max\{(\nu-\beta)^+,br_{ld}\}+\beta \\
&&\text{s.t. } (br_{ld}-(\nu-\beta)^+)^+<bM_* .
   \end{IEEEeqnarray}
If $br_{ld} \!<\!(\nu-\beta)^+$, the minimum is achieved by any $0\!\leq\! \beta\!<\! x-r_{ld}b$, and thus $\Delta^1_{j}(r_{ld})= \nu$ for $\nu > br_{ld}$.  If $br_{ld} \geq(\nu-\beta)^+$, then
  \begin{IEEEeqnarray}{rCl}\label{eq:Delta1jPart}
\Delta^1_{j}(r_{ld})&=&\inf_{\beta\geq0} br_{ld}+\beta \\
&&\text{s.t. } br_{ld}-bM_*< (\nu-\beta)^+\leq br_{ld} .
   \end{IEEEeqnarray}
If $\beta> \nu$, the problem is minimized by $\beta^*= \nu +\epsilon$, $\epsilon>0$, and $\Delta_{ld}(r_{ld})=br_{ld}+ \nu +\epsilon$, for $r_{ld}\leq M_*$. For $0\leq \beta\leq  \nu$, we have $\beta^*=(x-r_{ld}b)^+$, and $\Delta^1_{j}(r_{ld})=\max\{br_{ld}, \nu \}$ if $br_{ld}\leq bM_*+ \nu$. Putting all these together, we obtain
\begin{IEEEeqnarray}{rCl}\label{eq:Delta1jSol}
\Delta^1_{j}(r_{ld})=\max\{br_{ld}, \nu \}\quad \text{if } br_{ld}\leq
 \nu +bM_*.\end{IEEEeqnarray}
 If $br_{ld}>  \nu +bM_*$, $\mathcal{A}_j^c$ is empty, and there is always
outage.

 Next we solve the second optimization problem in (\ref{eq:Delta2j}). 
With $\beta= \nu$, $\Delta^2_{j}(r_{ld})$ is minimized and the range of
$\boldsymbol\alpha$ is enlarged. Then, the problem to solve reduces to
\begin{IEEEeqnarray}{rCl}
\Delta^2_{j}(r_{ld})&=&\inf  \nu +S(\boldsymbol\alpha) \\
&\text{s.t. }& r_{ld}\geq \sum_{i=1}^{M_*}(1-\alpha_i)^+, 
\end{IEEEeqnarray}
which is the DMT problem in (\ref{eq:DeltaBoundProblem2}). Hence,
$\Delta_{2j}(r_{ld},b)= \nu +d^{*}(r_{ld})$. Bringing all together,
\begin{IEEEeqnarray}{rCl}\label{eq:JointProblem}
\Delta_{ld}(b,\nu)=\max_{r_{ld}\geq0}\{\min\{\max\{x,br_{ld}\}, \nu +d^*(r_{ld})\}\}.
\end{IEEEeqnarray}
Since $d^*(r_{ld})=0$ for $r_{ld}> M_*$, the constraint in (\ref{eq:JointProblem}) can be reduced to $0\leq r_{ld}\leq M_*$ without loss of optimality since $\Delta_{ld}(b,\nu)= \nu$ for any $r_{ld} \geq M_*$. 
Then, the maximum achieved when the two terms inside $\min\{\cdot\}$ are
equal, i.e., $\max\{br_{ld},  \nu \}= \nu +d^*(r_{ld})$.  We chose a rate $r_{ld}$ such that
$br_{ld}> \nu$ and $r_{ld}<M_*$, as otherwise, the solution is readily
given by $\Delta_{ld}(b,\nu)= \nu$. Note that for $bM_*\leq  \nu$ this is never feasible, and thus, $\Delta_{ld}(b,\nu)= \nu$, and if $\nu \geq b \cdot d^*(M_*)$, the
intersection is always at $b r_{ld}= \nu$. Assuming $k\leq r_{ld}\leq k+1$,
$k=0,...,M_*-1$, the optimal $r_{ld}$ satisfies
at $br_{ld}=d^*(r_{ld})+ \nu$, or, equivalently,  $b r_{ld}= \nu +\Phi_k-(r_c-k)\Upsilon_k$, and we have
\begin{IEEEeqnarray}{rCl}
r_{ld}^*=\frac{\Phi_k+k\Upsilon_k+\nu }{\Upsilon_k+b}, \quad \Delta_{ld}(b,\nu)=br^*_{j}=b\frac{\Phi_k+k\Upsilon_k+\nu }{\Upsilon_k+b}. \label{eq:DistSSCCaux}
\end{IEEEeqnarray}
Since solution $r_{ld}^*$ is feasible whenever $k< r^*_{j}\leq k+1$, this solution is defined in
\begin{IEEEeqnarray}{rCl}
b\in\left[\frac{\Phi_{k+1}+\nu }{k+1},
\frac{\Phi_k+\nu }{k}\right), \quad \text{for } k=0,...,M_*-1,
\end{IEEEeqnarray}
where we have used $\Phi_{k+1}=\Phi_{k}-\Upsilon_k$. Notice that, whenever $\Delta_{ld}(b,\nu)\leq  \nu$ in (\ref{eq:DistSSCCaux}), we have $br^*_{j}\leq  \nu$, which is not feasible, and therefore $\Delta_{ld}(b,\nu)= \nu$. Remember that for $bM_*\leq  \nu$ we also have $\Delta_{ld}(b,\nu)= \nu$.
Putting all these cases together completes the proof of Theorem \ref{th:ExponentSeparation}.

\section{Proof Theorem \ref{th:DistortionExponentHDAExponent}}
\label{app:DistortionExponentHDAExponent}
In this Appendix we derive 
 the distortion exponent achieved by HDA-LD.
The outage region in (\ref{eq:OutageRegionHDASec}) is given by
\begin{IEEEeqnarray}{rCl}\label{eq:OutageCondHDA}
\mathcal{O}_{h}\!=\!\Bigg{\{}(\mathbf{H},\gamma)\!&:& \! \left(1+\frac{1}{\sigma_Q^2}\right)^{M_*} \\
&&\!\geq\!\Big{(}((1+\rho_s\gamma)(1+\sigma_Q^2))^{M_*}
\cdot\!\left.\frac{\prod_{i=1}^{M_*}(1+\frac{\rho}{M_*}\lambda_i)^{bM_*}}{\prod_{i=1}^{M_*}(1+\frac{\rho}{M_*}\lambda_i+(1+\rho_s\gamma)\sigma^2_Q)}\right)\Bigg{\}}.
\end{IEEEeqnarray}

Similarly to the analysis of the previous schemes, we consider the change
of variables $\lambda_i=\rho^{-\alpha_i}$, and $\gamma=\rho^{-\beta}$, and a rate $R_h=r_h\log\rho$, for $r_h\geq 0$. Then, we start by finding the equivalent outage set in the high SNR regime. We have,
\begin{IEEEeqnarray}{rCl}
&\prod_{i=1}^{M_*}&\left(1+\frac{\rho}{M_*}\lambda_i\right)^{bM_*}\doteq \rho^{bM_*\sum_{i=1}^{M_*}(1-\alpha_i)^+}, 
\end{IEEEeqnarray}
and
\begin{IEEEeqnarray}{rCl}
\prod_{i=1}^{M_*}\left(1+\frac{\rho}{M_*}\lambda_i+(1+\rho_s\gamma)\sigma^2_Q\right) 
&\doteq& \prod_{i=1}^{M_*}\left(1+\rho^{1-\alpha_i}+(1+\rho^{\nu-\beta})\rho^{-r_h}\right) \\
&\doteq&\rho^{\sum_{i=1}^{M_*}\max\{(1-\alpha_i)^+,(\nu-\beta)^+-r_h\}} ,
\end{IEEEeqnarray}
where we use $\sigma_Q^2=(2^{R_{h}-\epsilon}-1)^{-1}=(2^{-\epsilon}\rho^{r_{h}}-1)^{-1}\doteq\rho^{-r_h}$.
For the outage condition in (\ref{eq:OutageCondHDA}), we have
\begin{IEEEeqnarray}{lCl}
\frac{\left(1+\frac{1}{\sigma_Q^2}\right)^{M_*}\prod_{i=1}^{M_*}(1+\frac{\rho}{M_*}\lambda_i+(1+\rho_s\gamma)\sigma^2_Q) 
}{((1+\rho_s\gamma)(1+\sigma_Q^2))^{M_*}\prod_{i=1}^{M_*}(1+\frac{\rho}{M_*}\lambda_i)^{bM_*}} \\
\doteq\frac{\rho^{M_* r_{h}}\rho^{\sum_{i=1}^{M_*}\max\{(1-\alpha)^+,(\nu-\beta)^+-r_h\}}}{\rho^{M_*(\nu-\beta)^+}\rho^{bM_*\sum_{i=1}^{M_*}(1-\alpha)^+}} \\
\doteq\rho^{\sum_{i=1}^{M_*}(r_h-(\nu-\beta)^++(1-\alpha_i))^+-bM_*\sum_{1}^{M_*}(1-\alpha_i)^+}. 
\end{IEEEeqnarray}

Therefore, in the high SNR regime, the set $\mathcal{O}_h$ is equivalent to the set given by
\begin{IEEEeqnarray}{lCl}
\mathcal{A}_{h} \triangleq\Bigg{\{}(\boldsymbol\alpha,\beta)^+&:&\sum_{i=1}^{M_*}(r_{h}-(\nu-\beta)^++(1-\alpha_i))^+ 
>bM_*\sum_{i=1}^{M_*}(1-\alpha_i)\Bigg{\}}. 
\end{IEEEeqnarray}
On the other hand, in the high SNR regime, the distortion achieved by HDA-LD is equivalent to
\begin{IEEEeqnarray}{rCl}
D_{h}(\sigma_Q^2,\mathbf{H},\gamma)&=&\frac{1}{M_*}\sum_{i=1}^{M_*}\left(
1+\rho_s\gamma+\frac{1}{\sigma^2_Q}\left(1+\frac{\rho}{M_*}\lambda_i\right)\huge\right)^{-1} \\
&\doteq&\sum_{i=1}^{M_*}\left(
1+\rho^{\nu-\beta}+\rho^{r_h+(1-\alpha_i)}\huge\right)^{-1} \\
&\doteq&\rho^{-\min_{i=1,...,M_*}\{\max\{(\nu-\beta)^+,r_h+1-\alpha_i\}\}} \\
&\doteq&\rho^{-\max\{(\nu-\beta)^+, r_h+1-\alpha_{1}\}} , 
\end{IEEEeqnarray}
where the last equality follows since $\alpha_{1}\geq...\geq\alpha_{M_*}\geq 0$.
Then, in the high SNR regime, the expected distortion for HDA-LD is given as
\begin{IEEEeqnarray}{rCl}
ED_h(r_h\log\rho) 
&=&\int_{\mathcal{O}_h^c} D_h(\sigma^2_{Q},\mathbf{H},\gamma) p_{h}(\mathbf{H})p_{\Gamma}(\gamma)d \mathbf{H} d \gamma \\
&&+ \int_{\mathcal{O}_h} D_{d}(0,\gamma) p_{h}(\mathbf{H})p_{\Gamma}(\gamma)d \mathbf{H} d \gamma \\
&\doteq&\int_{\mathcal{A}_j^c} \rho^{-\max\{(\nu-\beta)^+,r_h+(1-\alpha_{1})\}} p_{A}(\boldsymbol\alpha)p_{B}(\beta)d \boldsymbol\alpha d \beta \\
&&+ \int_{\mathcal{A}_j} \rho^{-(\nu-\beta)^+}p_{A}(\boldsymbol\alpha)p_{B}(\beta) d \boldsymbol\alpha d \beta . 
\end{IEEEeqnarray}

Similarly to the proof of Theorem~\ref{the:ExponentJoint}, applying Varadhan's lemma, the exponent of each integral is found as
\begin{IEEEeqnarray}{lCl}\label{eq:Delta1HDA}
\Delta^1_{h}(r_{h})&=&\inf_{\mathcal{A}_{h}^c}\max\{(\nu-\beta)^+,r_{h}+1-\alpha_{1}\}+S_A(\boldsymbol\alpha)+\beta,  
\end{IEEEeqnarray}
and
\begin{IEEEeqnarray}{lCl}\label{eq:Delta2HDA}
\Delta^2_{h}(r_{h})&=&\inf_{\mathcal{A}_{h}}(\nu-\beta)^++S_A(\boldsymbol\alpha)+\beta,
\end{IEEEeqnarray}

First we solve $\Delta^1_{h}(r_{h})$. The infimum for this problem is
achieved by $\boldsymbol\alpha^*=0$ and $\beta^*=0$, and is given by
\begin{IEEEeqnarray}{rCl}
\Delta^1_{h}(r_{h}) =\max\{x,r_{h}+1\}, &\qquad\text{for
}r_{h}\leq M_* b-1+\nu . 
\end{IEEEeqnarray}
Now we solve $\Delta^2_h(r_h)$ in (\ref{eq:Delta2HDA}). By letting $\beta^*= \nu$,
the range of $\boldsymbol\alpha$ is enlarged while the objective function is minimized. Thus, the problem  reduces to
\begin{IEEEeqnarray}{rCl}\label{eq:Delta1Hybbg2}
\Delta^2_{h}(r_{h})&=&\inf  \nu +S(\boldsymbol\alpha) \\
&\text{s.t. }&  r_{h} >\frac{bM_*-1}{M^*}\sum_{i=1}^{M_*}(1-\alpha_i)^+. 
\end{IEEEeqnarray}
Again, this problem is a scaled version of the DMT curve in (\ref{eq:DeltaBoundProblem2}). Therefore, we have
\begin{IEEEeqnarray}{rCl}
\Delta^2_{h}(r_{h})= \nu +d^*\left(\left(\frac{bM_*-1}{M_*}\right)^{-1}r_{h}\right). 
\end{IEEEeqnarray}

The distortion exponent is given by optimizing over $r_h$ as
\begin{IEEEeqnarray}{rCl}
\Delta_{h}(b,\nu)=\max_{r_{h}}\min\{\Delta^1_{h}(r_{h}),\Delta^2_{h}(r_{h})\}. 
 \end{IEEEeqnarray}
The maximum distortion exponent is obtained by
letting $\Delta^1_h(r_{h})=\Delta^2_h(r_{h})$. We assume $r_{h}+1> \nu$ since otherwise $\Delta_{h}(b,\nu)= \nu$, and then, we have
$r_{h}+1= \nu +d^*\left((b-\frac{1}{M_*})^{-1}r_{h}\right)$. Let $r_{h}'=r_{h}(b-\frac{1}{M_*})^{-1}$. Using (\ref{eq:DMTcurve}), for $k<r_{h}'\leq k+1$, $k=0,...,M_*-1$, the problem is equivalent to $r_{h}'\left(b-\frac{1}{M_*}\right)+1= \nu +\Phi_k-(r_h'-k)\Upsilon_k$, where $\Phi_k$ and $\Upsilon_k$ are given as in (\ref{eq:ML_PhiUpsilon}).  The $r_h'$ satisfying the equality is given by
\begin{IEEEeqnarray}{rCl}
r_h'^{*}=\frac{\Phi_k+k\Phi_k-1+\nu }{b-\frac{1}{M_*}+\Phi_k}, 
\end{IEEEeqnarray}
and the corresponding distortion exponent is found as
\begin{IEEEeqnarray}{rCl}
\Delta_{h}(b,\nu)=1+\frac{(bM_*-1)
(\Phi_k+k\Upsilon_k-1+\nu )}{bM_*-1+M_*\Upsilon_k}, 
\end{IEEEeqnarray}
for
\begin{IEEEeqnarray}{rCl}
b&\in&\left[\frac{\Phi_{k+1}-1+\nu }{k+1}+\frac{1}{M_*},\frac{\Phi_k-1+\nu }{k}+\frac{1}{M_*}\right), \qquad\text{ for }k=0,...,M_*-1.
\end{IEEEeqnarray}
Note that we have $r^*_h+1> \nu$ whenever $\Delta_{h}(b,\nu)> \nu$. Otherwise, $r_h^*$ is not feasible and $\Delta_{h}(b,\nu)= \nu$. Note also that if $\nu \geq bM_*$, the distortion exponent is given by
$\Delta_{h}(b,\nu)= \nu$.

\section{Proof of Theorem \ref{th:DistExpLSLD}}\label{app:ProgressiveLD}

%
In this section we obtain the distortion exponent for LS-LD. Let us define $\bar{R}_{1}^l\triangleq \sum_{i=1}^{l}R_i$. First, we consider the outage event. For the successive refinement codebook  the l.h.s. of (\ref{eq:ProgLDOutCond}) is given by
\begin{IEEEeqnarray}{rCl}
I(S;W_l|W_1^{l-1},Y)&\stackrel{(a)}{=}&I(S;W_l|Y)-I(S;W_{l-1}|Y) \\
&\overset{(b)}=&H(W_l|Y)-H(\overline{Q}_l)-H(W_{l-1}|Y)+H(\overline{Q}_{l-1}) \\
&\stackrel{(c)}{=}&\log\left(\frac{\sum_{i=l-1}^L\sigma_{i}^2}{\sum_{i=l}^L\sigma_{i}^2}\frac{1+(1+\gamma\rho_s)\sum_{j=l}^L\sigma_j^2}{1+(1+\gamma\rho_s)\sum_{j=l-1}^L\sigma_j^2}\right),\label{eq:InfExp} 
\end{IEEEeqnarray}
where $\overline{Q}_l \triangleq \sum_{i=l}^L Q_l$, and
 $(a)$ is due to the Markov chain $T-S-W_L-...-W_1$, and $(b)$ is due to the independence of $\bar{Q}_i$ from $S$
and $T$, and finally $(c)$ follows since $H(W_l|T)=\frac{1}{2}\log\left(\sum_{i=l}^L\sigma_{i}^2+\frac{1}{1+\gamma\rho_s}\right)$ for $l=1,...,L$. We also have
\begin{IEEEeqnarray}{rCl}
&I&(S;W_1|T)=\log\left(1+\frac{1}{(1+\gamma\rho_s)\sum_{i=1}^L\sigma_{i}^2}\right). 
\end{IEEEeqnarray}

Substituting (\ref{eq:Vark}) into (\ref{eq:InfExp}), we have
\begin{IEEEeqnarray}{rCl}\label{eq:RHS1}
I(S;W_l|W_1^{l-1},T)=\log\left(\frac{2^{\sum_{i=1}^l\frac{b}{L}R_i-\epsilon}+\gamma\rho_s}{2^{\sum_{i=1}^{l-1}\frac{b}{L}R_i-\epsilon}+\gamma\rho_s}\right) .
\end{IEEEeqnarray}

Then, the outage condition in (\ref{eq:ProgLDOutCond}) is given by
\begin{IEEEeqnarray}{rCl}
\log\left(\frac{2^{\sum_{i=1}^l\frac{b}{L}R_i-\epsilon}+\gamma\rho_s}{2^{\sum_{i=1}^{l-1}\frac{b}{L}R_i-\epsilon}+\gamma\rho_s}\right)\geq \frac{b}{L}\log \prod_{i=1}^{M_*}\left(1+\frac{\rho}{M_*}\lambda_i\right).
\end{IEEEeqnarray}

Therefore, in the high SNR regime, we have, for $l=1,...,L$
\begin{IEEEeqnarray}{rCl}\label{eq:ML_OutageRight}
\frac{2^{\sum_{i=1}^l(\frac{b}{L}R_i-\epsilon)}+\gamma\rho_s}{2^{\sum_{i=1}^{l-1}(\frac{b}{L}R_i-\epsilon)}+\gamma\rho_s}
&\doteq&\frac{\rho^{\sum_{i=1}^l\frac{b}{L}r_i}+\rho^{\nu-\beta}}{\rho^{\sum_{i=1}^{l-1}\frac{b}{L}r_i}+\rho^{\nu-\beta}}\\
&\doteq&\frac{\rho^{\sum_{i=1}^l\frac{b}{L}r_i-(\nu-\beta)}+1}{\rho^{\sum_{i=1}^{l-1}\frac{b}{L}r_i-(\nu-\beta)}+1} \\
&\doteq&\frac{\rho^{(\sum_{i=1}^l\frac{b}{L}r_i-(\nu-\beta))^+}}{\rho^{(\sum_{i=1}^{l-1}\frac{b}{L}r_i-(\nu-\beta))^+}}, 
\end{IEEEeqnarray}
and
\begin{IEEEeqnarray}{rCl}
\frac{b}{L}\log \prod_{i=1}^{M_*}\left(1+\frac{\rho}{M_*}\lambda_i\right)\doteq\rho^{\frac{b}{L}\sum_{i=1}^{M_*}(1-\alpha_i)^+}. 
\end{IEEEeqnarray}
The outage set (\ref{eq:ProgLDOutCond}) in the high SNR regime is equivalent to
\begin{IEEEeqnarray}{rCl}\label{eq:OutageProgrLD}
\mathcal{A}^{ls}_{l}&\triangleq&\left\{(\boldsymbol\alpha,\beta):\frac{b}{L}\sum_
{i=1}^{M_*}[\left(1-\alpha_i\right)^+\right.
<\left(\sum_{i=1}^l\frac{b}{L}r_i-(\nu-\beta)\right)^+
\!-\!\left.\left(\sum_{i=1}^{l-1}\frac{b}{L}r_i-(\nu-\beta) \right)^+\right\}.
\end{IEEEeqnarray}

Now, we study the high SNR behavior of the expected distortion.
It is not hard to see that (\ref{eq:EDLS}) is given by
\begin{IEEEeqnarray}{rCl}\label{eq:EDLSdec}
ED_{ls}(\mathbf{R})\!&=&\!\sum_{l=0}^{L}\mathrm{E}_{\mathcal{O}^{ls}_{l+1}}\left[D_{d}\left(\frac{b}{L}\bar{R}_{1}^l,\gamma\right)\right]-\mathrm{E}_{\mathcal{O}^{ls}_{l}}\left[D_{d}\left(\frac{b}{L}\bar{R}_{1}^l,\gamma\right)\right],
\end{IEEEeqnarray}
where $\mathcal{O}_{0}^{ls}\triangleq \emptyset$ and $\mathcal{O}_{L+1}^{ls}\triangleq \mathds{R}^{M_*+1}$.
For each term in (\ref{eq:EDLSdec}), we have
\begin{IEEEeqnarray}{lCl}
\mathrm{E}_{\mathcal{O}^{ls}_{l+1}}\left[D_{d}\left(\frac{b}{L}\bar{R}_{1}^l,\gamma\right)\right]
&\doteq&\int_{\mathcal{A}^{ls}_{l+1}}\rho^{-\max\{\frac{b}{L}\sum_{i=1}^{l}r_l,(\nu-\beta)^+\}}\rho^{-S_{A}(\boldsymbol\alpha)}\rho^{-\beta}d\boldsymbol\alpha d\beta,\label{eq:LSTems1}\\
\mathrm{E}_{\mathcal{O}^{ls}_{l}}\left[D_{d}\left(\frac{b}{L}\bar{R}_{1}^l,\gamma\right)\right]
&\doteq&\int_{\mathcal{A}^{ls}_{l}}\rho^{-\max\{\frac{b}{L}\sum_{i=1}^{l}r_l,(\nu-\beta)^+\}}\rho^{-S_{A}(\boldsymbol\alpha)}\rho^{-\beta}d\boldsymbol\alpha d\beta,\label{eq:LSTems2}
\end{IEEEeqnarray}
where the outage set in the high SNR regime is given by (\ref{eq:OutageProgrLD}).

Applying Varadhan's lemma to (\ref{eq:LSTems1}), the exponential behavior of (\ref{eq:LSTems1}) for $l=0,...,L-1$, is found as the solution to
\begin{IEEEeqnarray}{rCl}
\tilde{\Delta}^+_l\triangleq \inf_{\mathcal{A}^{ls}_{l+1}}\max\{b/L\bar{r}_1^{l},(\nu-\beta)^+\}+S_A(\boldsymbol\alpha)+\beta ,
\end{IEEEeqnarray}
where we define $\bar{r}_1^{l}\triangleq \sum_{i=1}^{l}r_i$. Similarly, applying Varadhan's lemma to (\ref{eq:LSTems2}),  the exponential behavior of (\ref{eq:LSTems2}) for $l=0,...,L-1$ is given by
\begin{IEEEeqnarray}{rCl}
\tilde{\Delta}_l\triangleq \inf_{\mathcal{A}^{ls}_{l}}\max\{b/L\bar{r}_1^{l},(\nu-\beta)^+\}+S_A(\boldsymbol\alpha)+\beta .
\end{IEEEeqnarray}

Since $r_1\leq r_2\leq \cdots\leq r_L $ we have $\mathcal{A}^{ls}_{l}\subseteq\mathcal{A}^{ls}_{l+1}$, and therefore $\tilde{\Delta}_{l}\geq \tilde{\Delta}^{+}_l$. Then, from (\ref{eq:EDLSdec}) we have
\begin{IEEEeqnarray}{rCl}
ED_{ls}(\mathbf{R})&\doteq&\sum_{l=0}^{L}\rho^{-\tilde{\Delta}^+_l}-\rho^{-\tilde{\Delta}_l}\doteq\!\sum_{l=0}^{L}\rho^{-\Delta^+_l} .
\end{IEEEeqnarray}
We define $\Delta^{ls}_l(\mathbf{r})\triangleq\tilde{\Delta}^+_l$, where $\mathbf{r}\triangleq [r_1,...,r_L]$. Then, the distortion exponent of LS-LD is given as follows:
\begin{IEEEeqnarray}{rCl}
\Delta^*_{ls}(b,\nu)=\max_{\mathbf{r}}\min{\Delta^{ls}_l}(\mathbf{r}). 
\end{IEEEeqnarray}

For $l\!\!=\!0$, i.e., no codeword is successfully decoded, we have
\begin{IEEEeqnarray}{rCl}
\Delta^{ls}_0(\mathbf{r})&=&\inf
 (\nu-\beta)^++\beta+S_A(\boldsymbol\alpha) \\
&&\text{s.t. }\,\frac{b}{L}\sum_{i=1}^{M_*}(1- \alpha_i)^+<\left(\frac{b}{L}r_1-(\nu-\beta)\right)^+ .
\end{IEEEeqnarray}
The infimum is achieved by $\beta= \nu$ and using the DMT in  (\ref{eq:DeltaBoundProblem2}), we have
\begin{IEEEeqnarray}{rCL}
\Delta^{ls}_0(\mathbf{r})&=& \nu +d^*\left(r_1\right). 
\end{IEEEeqnarray}
The distortion exponent when $l$ layers are successfully decoded is found as
\begin{IEEEeqnarray}{rCl}\label{eq:Problem0LS}
\Delta^{ls}_l(\mathbf{r})&=&\inf\max\left\{\frac{b}{L}\bar{r}_1^l,(\nu-\beta)^+\right\}+\beta+S_{A}(\boldsymbol\alpha) \\
&\text{s.t. }& 
\frac{b}{L}\sum_{i=1}^{M_*}\left(1-\alpha_i\right)^+<\left(\frac{b}{L}\bar{r}_1^{l+1}-(\nu-\beta)\right)^+-\left(\frac{b}{L}\bar{r}_1^{l}-(\nu-\beta)\right)^+. 
\end{IEEEeqnarray}

If $\frac{b}{L}\bar{r}_1^l\geq  \nu$, the infimum of (\ref{eq:Problem0LS}) is obtained for $\beta^*=0$ and
 \begin{IEEEeqnarray}{rCl}\label{eq:Problem1LS}
\Delta^{ls}_l(\mathbf{r})&=&\inf\frac{b}{L}\bar{r}_1^l+S_A(\boldsymbol\alpha) \\
&\text{s.t. }&\sum_{i^1}^{M_*}\left(\xi_{l}-\alpha_i\right)^+<r_{k+1} .
\end{IEEEeqnarray}
Using the DMT in (\ref{eq:DeltaBoundProblem2}), (\ref{eq:Problem1LS}) is minimized as
\begin{IEEEeqnarray}{rCl}
\Delta^{ls}_l(\mathbf{r})=\frac{b}{L}\bar{r}_{1}^l+d^*\left(r_{l+1}\right). 
\end{IEEEeqnarray}

If $\frac{b}{L}\bar{r}_1^l\leq  \nu$, we have that the minimum of (\ref{eq:Problem0LS}) is achieved by $\beta^*=\left(x-\frac{b}{L}\bar{r}_1^l\right)^+$ if $\frac{b}{L}\bar{r}_1^{l}>(\nu-\beta)$ and is given by
\begin{IEEEeqnarray}{rCl}
\Delta^{ls}_l(\mathbf{r})= \nu +d^*\left(r_{l+1}\right). 
\end{IEEEeqnarray}
If $\frac{b}{L}\bar{r}_1^{l}\leq(\nu-\beta)<\frac{b}{L}\bar{r}_1^{l+1}$, the optimization problem in (\ref{eq:Problem0LS}) is equivalent to
\begin{IEEEeqnarray}{rCl}\label{eq:Problem3LS}
\Delta^{ls}_l(\mathbf{r})&=&\inf
 (\nu-\beta)^++\beta+S_A(\boldsymbol\alpha)\\
&&\text{s.t. }\,\frac{b}{L}\sum_{i=1}^{M_*}(1- \alpha_i)^+<\left(\frac{b}{L}r_1^{l+1}-(\nu-\beta)\right)^+ ,\\
&&\qquad\frac{b}{L}\bar{r}_1^{l}\leq(\nu-\beta)<\frac{b}{L}\bar{r}_1^{l+1}. 
\end{IEEEeqnarray}
The infimum of (\ref{eq:Problem3LS}) is achieved by the largest $\beta$, since increasing $\beta$ enlarges the range of $\boldsymbol\alpha$. Then, $\beta^*=(x-\frac{b}{L}\bar{r}_1^l)^+$, and we have,
\begin{IEEEeqnarray}{rCl}
\Delta^{ls}_l(\mathbf{r})= \nu +d^*\left(r_{l+1}\right). 
\end{IEEEeqnarray}

Finally, if $\frac{b}{L}\bar{r}_1^{l+1}\leq(\nu-\beta)$, there are no feasible solutions for (\ref{eq:Problem0LS}). Therefore, putting all together we have
\begin{IEEEeqnarray}{rCl}
\Delta^{ls}_l(\mathbf{r})&=&\inf\max\left\{\frac{b}{L}\bar{r}_1^l,x\right\}+d^*(r_{l+1}).   \end{IEEEeqnarray}

Similarly, at layer $L$, the infimum is achieved by $\boldsymbol\alpha^*=0$ and $\beta^*=0$ and is given by
\begin{IEEEeqnarray}{rCl}
\Delta^{ls}_L(\mathbf{r})&=&\max\left\{\frac{b}{L}\bar{r}_1^L,x\right\}, \quad \text{for } r_L\leq M_*. 
\end{IEEEeqnarray}
Note that the condition on $r_L$ always holds.
\subsection{Solution of the distortion exponent}
Assume that for a given layer $\hat{l}$ we have $\bar{r}^{\hat{l}-1}_1\frac{b}{L}\leq x\leq\bar{r}^{\hat{l}}_1\frac{b}{L}$. Then, $\Delta^{ls}_{l}(\mathbf{r})= \nu +d(r_{l+1})$ for $l=0,...,\hat{l}-1$. Using the KKT conditions, the maximum distortion exponent is obtained when all the distortion exponents are equal.

From $\Delta^{ls}_0(\mathbf{r})=\cdots=\Delta^{ls}_{\hat{l}-1}(\mathbf{r})$ we have $r_1=\cdots=r_{\hat{l}}$, and thus, $\bar{r}_{1}^{\hat{l}}=\hat{l}r_1$. Then, the exponents are given by
\begin{IEEEeqnarray}{rCl}
\Delta^{ls}_{0}(\mathbf{r})&=& \nu +d^*(r_{1}) \\
\Delta^{ls}_{\hat{l}}(\mathbf{r})&=&b\frac{\hat{l}}{L}r_1+d^*(r_{\hat{l}+1}) \\
&&\cdots \\
\Delta^{ls}_{L-1}(\mathbf{r})&=&b\frac{\hat{l}}{L}r_1+b\frac{1}{L}\bar{r}_{\hat{l}+1}^{L-1}+d^*(r_{L}) \\
\Delta^{ls}_{L}(\mathbf{r})&=&b\frac{\hat{l}}{L}r_1+b\frac{1}{L}\bar{r}_{\hat{l}+1}^{L}. 
\end{IEEEeqnarray}

Equating all these exponents, we have
\begin{IEEEeqnarray}{rCl}
b\frac{1}{L}r_{L}&=&d^*(r_L) \\
b\frac{1}{L}r_{L-1}+d(r_L)&=&d^*(r_{L-1}) \\
&\cdots& \\
b\frac{1}{L}r_{\hat{l}+1}+d^*(r_{\hat{l}+2})&=&d^*(r_{\hat{l}+1}) \\
b\frac{l}{L}r_1+d^*(r_{\hat{l}+1})&=&d^*(r_{1})+\nu . 
\end{IEEEeqnarray}

A geometric interpretation of the rate allocation for LS-LD satisfying the above equalities is the following: we have $L-\hat{l}$ straight lines of slope $b/L$ and each line intersects in the $y$ axis at a point with the same ordinate as the intersection of the previous line with the DMT curve. The more layers we have the higher the distortion exponent of LS-LD can climb. The remaining $\hat{l}$ layers allow a final climb of slope $\hat{l}b/L$.  Note that the higher $\hat{l}$, the higher the slope, but the lower the starting point $d^*(r_{\hat{l}+1})$.

Next, we adapt Lemma 3 from \cite{gunduz2008joint} to our setup. Let $q$ be a line with equation $y=-\alpha(t-M)$ for some $\alpha>0$ and $M>0$ and let $q_i=1,...,L$ be the set of lines defined recursively from $L$ to $1$ as $y=(b/L)t+d_{i+1}$, where $b>0$, $d_{L+1}\triangleq 0$, and $d_i$ is the $y$ component of the intersection of $q_i$ with $q$.
Then, sequentially solving the intersection points for $i=\hat{l}+1,...,L$ we have:
\begin{IEEEeqnarray}{rCl}
d_{i}-d_{i+1}=M\frac{b}{L}\left(\frac{\alpha}{\alpha+b/L}\right)^{L-i+1}. 
\end{IEEEeqnarray}
Summing all the terms for $i=\hat{l}+1,...,L$ we obtain
\begin{IEEEeqnarray}{rCl}
d_i=M\alpha\left[1-\left(\frac{\alpha}{\alpha+b/L}\right)^{L-i+1}\right]. 
\end{IEEEeqnarray}

In the following we consider a continuum of layers, i.e., we let $L\rightarrow \infty$. Let $\hat{l}=\kappa L$ be the numbers of layers needed so that $b\hat{l}/L r_1=b\kappa r_1= \nu$, that is, from $l=1$ to $l=\kappa L$.

When $M_*=1$, the DMT curve is composed of a single line with $\alpha=M^*$ and $M=1$. In that case, with layers from $\kappa L+1$ to $L$ the distortion exponent increases up to
\begin{IEEEeqnarray}{rCl}
d^*(r_{L\kappa+1})&=&M\alpha\left[1-\left(\frac{\alpha}{\alpha+b/L}\right)^{L(1-\kappa)}\right]. 
\end{IEEEeqnarray}
In the limit of infinite layers, we obtain
\begin{IEEEeqnarray}{rCl}
\lim_{L\rightarrow\infty}d^*(r_{L\kappa+1})=M \alpha \left(1-e^{-\frac{b(1-\kappa)}{\alpha}}\right). 
\end{IEEEeqnarray}

We still need to determine the distortion achieved due to the climb with layers from $l=1$ to $l=\kappa L$ by determining $r_1$, which is found as the solution to $\Delta^{ls}_0(\mathbf{r})=\Delta^{ls}_{L\kappa}(\mathbf{r})$, i.e.,
\begin{IEEEeqnarray}{rCl}\label{eq:LS_SIMO_equation}
b\kappa r_1+d^*(r_{L\kappa+1})&=&\nu-\alpha(r_1-M),
\end{IEEEeqnarray}

Since $\nu =b\kappa r_1$, $r_1=x/b\kappa$, and from (\ref{eq:LS_SIMO_equation}) we get to
\begin{IEEEeqnarray}{rCl}
d^*(r_{L\kappa+1})&=&-\alpha\left(\frac{ \nu}{b\kappa}-M\right), 
\end{IEEEeqnarray}
which, in the limit of infinite layers, solves for
\begin{IEEEeqnarray}{rCl}
\kappa^*=\frac{M^*}{b}\mathcal{W}\left(\frac{e^{\frac{b}{M^*}x}}{M^*}\right), 
\end{IEEEeqnarray}
 where $\mathcal{W}(z)$ is the Lambert $W$ function, which gives the principal solution for $w$ in $z=we^{w}$.
The distortion exponent in the MISO/SIMO case is then found as
\begin{IEEEeqnarray}{rCl}
\Delta^*_{ls}(b,\nu)= \nu +M^*\left(1-e^{-\frac{b(1-\kappa^*)}{M^*}}\right). 
\end{IEEEeqnarray}

For  MIMO channels, the DMT curve is formed by $M_*$ linear pieces, each between $M_*-k$ and $M_*-k+1$ for $k=1,...,M_*$. From the value of the DMT at $M_*-k$ to the value at $M_*-k+1$, there is a gap of $M^*-M_*+2k-1$ in the $y$ abscise. Each piece of the curve can be characterized by $y=-\alpha(t-M)$, where for the $k$-th interval we have $\alpha=\phi_k$ and $M=M_k$ as in (\ref{eq:LS_alpha_and_M}).

We will again consider a continuum of layers, i.e., we let $L\rightarrow\infty$, and we let $l=L\kappa$ be the number of lines required to have $b\kappa r_1= \nu$. Then, for the remaining lines from $l+1$ to $L$, let $L(1-\kappa)\kappa_k$ be the number of lines with slope $b/L$ required to climb up the whole interval $k$. Since the gap in the $y$ abscise from the value at $M_*-k$ to the value at $M_*-k+1$, is $M^*-M_*+2k-1$, climbing the whole $k$-th  interval with $L(1-\kappa)\kappa_k$ lines requires
\begin{IEEEeqnarray}{rCl}
d_{L-L(1-\kappa)\kappa_k}&=& M^*-M_*+2 k -1, 
\end{IEEEeqnarray}
where
\begin{IEEEeqnarray}{rCl}
d_{L-L(1-\kappa)\kappa_k}&=&M\alpha\left[1-\left(\frac{\alpha}{\alpha+b/L}\right)^{L(1-\kappa)\kappa_k+1}\right]. 
\end{IEEEeqnarray}

In the limit we have
\begin{IEEEeqnarray}{rCl}
\lim_{L\rightarrow\infty}d_{L-L(1-\kappa)\kappa_k}&=&M\alpha\left[1-e^{-\frac{b(1-\kappa)\kappa_k}{\alpha}}\right]. 
\end{IEEEeqnarray}

Then, each required portion, $\kappa_k$, is found as
\begin{IEEEeqnarray}{rCl}
\kappa_k=\frac{M^*-M_*+2 k -1}{b(1-\kappa)}\ln\left(\frac{M_*-k+1}{M_*-k}\right). 
\end{IEEEeqnarray}

This gives the portion of lines required to climb up the $k$
-th segment of the DMT curve. In the MIMO case, to be able to go up exactly to the $k$-th segment with lines from $l+1$ to $L$ we need to have
$\sum_{j=1}^{k-1}\kappa_j<1\leq\sum_{j=1}^{k}\kappa_j$.
This is equivalent to the requirement $c_{k-1}<b(1-\kappa)\leq c_k$ using $c_i$ as defined in Theorem \ref{th:DistExpLSLD}. To climb up each line segment we need $\kappa_k(1-\kappa)L$ lines (layers) for $k=1,...,M_*-1$, and for the last segment climbed we have $(1-\sum_{j=1}^{k-1}\kappa_j)L$ lines remaining,
which gives an extra ascent of
\begin{IEEEeqnarray}{rCl}
M\alpha\left(1-e^{-\frac{b(1-\kappa)(1-\sum_{j=1}^{k-1}\kappa_j)}{\alpha}}\right) .
\end{IEEEeqnarray}
Then, we have climbed up to the value
\begin{IEEEeqnarray}{rCl}
d_{L\kappa+1}&=&\sum_{i=1}^{k-1}(M^*-M_*+2 i-1 ) \\
&&+(M_*-k+1)(M^*-M_*+2 k -1) \left(1-e^{-\frac{b(1-\kappa)(1-\sum_{j=1}^{k-1}\kappa_j)}{M^*-M_*+2 k -1}}\right). 
\end{IEEEeqnarray}
With the remaining lines, i.e., from $l=1$ to $l=\kappa L$, the extra climb is given by solving $\Delta^{ls}_{0}(\mathbf{r})=\Delta^{ls}_{\kappa L}(\mathbf{r})$, i.e.,
\begin{IEEEeqnarray}{rCl}
 \nu +d^*(r_{1})&=&b\kappa r_1+d_{L\kappa+1}. 
\end{IEEEeqnarray}

The diversity gain $d^*(r_1)$  at segment $k$ is given by
\begin{IEEEeqnarray}{rCl}
d^*(r_1)=-\alpha(r_1-M)+\sum_{i=1}^{k-1}(M^*-M_*+2 i-1). 
\end{IEEEeqnarray}

Since we have $b\kappa r_1= \nu$,  this equation simplifies to
\begin{IEEEeqnarray}{rCl}
d^*\left(\frac{ \nu}{b\kappa}\right)&=&d_{L\kappa+1}. 
\end{IEEEeqnarray}

Therefore, using $c_{k-1}\triangleq b(1-\kappa)\sum_{j=1}^{k-1}\kappa_j$,  we solve $\kappa$ from
\begin{IEEEeqnarray}{rCl}
-\alpha\left(\frac{ \nu}{b\kappa}-M\right)&=&M\alpha\left(1-e^{-\frac{b(1-\kappa)-c_{k-1}}{\alpha}}\right) ,
\end{IEEEeqnarray}
and find
\begin{IEEEeqnarray}{rCl}
\kappa^*=\frac{\alpha}{b}\mathcal{W}\left(\frac{e^{\frac{b-c_{k-1}}{\alpha}}x}{M \alpha}\right). 
\end{IEEEeqnarray}

The range of validity for each
$k$ is given by $c_{k-1}<b(1-\kappa)\leq c_{k}$.
Since for a given $c$, the solution to $c=b(1-\kappa^*)$ is found as
\begin{IEEEeqnarray}{rCl}
b=\frac{xe^{c_{k-1}-c}}{M}+c, 
\end{IEEEeqnarray}
when $c=c_{k-1}$, we have
\begin{IEEEeqnarray}{rCl}
b>\frac{ \nu}{M}+c_{k-1}=c_{k-1}+\frac{ \nu}{M_*-k+1}. 
\end{IEEEeqnarray}
When $c=c_{k}$, since $c_{k-1}-c_k=\alpha\ln(M/(M_*-k))$, we have
\begin{IEEEeqnarray}{rCl}
b\leq \frac{\nu e^{c_{k-1}-c_k}}{M}+c_{k}=c_{k}+\frac{ \nu}{M_*-k}. 
\end{IEEEeqnarray}

Putting all together, we obtain the condition of the theorem and the corresponding distortion exponent.

\section{ Proof of Theorem \ref{the:MLfintieLayers} }\label{app:ExponentLayerLD_DistExpSol}

We consider the usual change of
variables, $\lambda_i=\rho^{-\alpha_i}$
and $\gamma=\rho^{-\beta}$. Let $r_l$ be the multiplexing gain of the $l$-th layer and $\mathbf{r}\triangleq[r_1,...,r_L]$, such that $R_i=r_i\log \rho$, and define $\bar{r}_{1}^l\triangleq \sum_{i=1}^lr_i$.

First, we derive the outage set $\mathcal{O}^{bs}_l$ for each layer in the high SNR regime, which we denote by $\mathcal{L}_l$.
For the power allocation $\rho_l=\rho^{\xi_{l-1}}-\rho^{\xi_l}$, the
l.h.s. of the inequality in the definition of $\mathcal{O}^{bs}_l$ in (\ref{eq:OutageJointMultiple}) is given by
\begin{IEEEeqnarray}{rCl}\label{eq:ML_OutageLeft}
I(\mathbf{X}_{l};\mathbf{Y}|\mathbf{X}_{1}^{l-1})&=&I(\mathbf{X}_{l}^{L};\mathbf{Y}|\mathbf{X}_{1}^{l-1})-I(\mathbf{X}_{l+1}^{L};\mathbf{Y}|\mathbf{X}_{1}^{l-1}) \\
&=&\log\frac{\det\left(\mathbf{I}+\frac{\rho^{\xi_{l-1}}}{M_*}\mathbf{HH}^H\right)}{\det\left(\mathbf{I}+\frac{\rho^{\xi_{l}}}{M_*}\mathbf{HH}^H\right)} \\
&=&\log\prod_{i=1}^{M_*}\frac{1+\frac{\rho^{\xi_{l-1}}}{M_*}\lambda_i}{1+\frac{\rho^{\xi_{l}}}{M_*}\lambda_i} \\
&\doteq&\rho^{\sum_{i=1}^{M_*}(\xi_{l-1}-\alpha_i)^+-(\xi_{l}-\alpha_i)^+}.
\end{IEEEeqnarray}

The r.h.s. of the inequality in the definition of $\mathcal{O}^{bs}_l$ in (\ref{eq:OutageJointMultiple}) can be calculated
as in (\ref{eq:ML_OutageRight}). Then, from (\ref{eq:ML_OutageLeft}) and (\ref{eq:ML_OutageRight}), $\mathcal{L}_l$ follows as:

\begin{IEEEeqnarray}{rCl}
\mathcal{L}_l&\triangleq&\left\{(\boldsymbol\alpha,\beta):b\sum_
{i=1}^{M_*}[\left(\xi_{l-1}-\alpha_i\right)^+-\left(\xi_{l}-\alpha_i\right)^+]\right. \\
&&\qquad\qquad<\left(\sum_{i=1}^lbr_i-(\nu-\beta)\right)^+
-\left.\left(\sum_{i=1}^{l-1}br_i-(\nu-\beta) \right)^+\right\}. 
\end{IEEEeqnarray}

Since $\mathcal{O}^{bs}_l$ are mutually
exclusive, in the high SNR we have
\begin{IEEEeqnarray}{rCl}
ED_{bs}(\mathbf{R},\boldsymbol\xi)  
&=&\sum_{l=0}^L\int_{\mathcal{O}^{bs}_{l+1}}D_d\left(\sum^l_{i=0}bR_i,\gamma\right)p_{h}(\mathbf{H})
p_{\Gamma}(\gamma)d\mathbf{H} d{\gamma} \\
&\doteq&\sum_{l=0}^L\!\int_{\mathcal{L}_{l+1}}\!\!\rho^{-(\max\{\sum^l_{i=0}br_i,(\nu-\beta)^+\}+\beta+S_A(\boldsymbol\alpha))}d\boldsymbol\alpha d{\beta} \\
&\doteq&\sum_{l=0}^L\rho^{-\Delta_l(\mathbf{r},\boldsymbol\xi)}  \\
&\doteq&\rho^{-\Delta_{bs}^{L}(\mathbf{r},\boldsymbol\xi)},
\end{IEEEeqnarray}
where, from Varadhan's lemma, the exponent for each integral term is given by
\begin{IEEEeqnarray}{rCl}\label{eq:MultilayerNBJD_DistExpProblemLayer}
\Delta^{bs}_l(\mathbf{r},\boldsymbol\xi)&=&\!\inf_{\mathcal{L}_{l+1}}\!
\max\left\{b\bar{r}_{0}^l,(\nu-\beta)^+\right\}+\beta+S_A(\boldsymbol\alpha).
\end{IEEEeqnarray}
Then, the distortion exponent is found as
\begin{IEEEeqnarray}{rCl}\label{eq:MultilayerNBJD_DistExpProblem}
\Delta_{bs}^L(b,\nu)=\max_{\mathbf{r},\boldsymbol\xi}\min_{l=0,...,L}
\left\{\Delta^{bs}_l(\mathbf{r},\boldsymbol\xi)\right\}.
\end{IEEEeqnarray}

Similarly to the DMT, we consider the \emph{successive decoding diversity gain}, defined in
\cite{gunduz2008joint}, as the solution to the
probability of outage with successive decoding of each layer, given by
\begin{IEEEeqnarray}{rCl}\label{eq:succDec1}
d_{ds}(r_l,\xi_{l-1},\xi_{l})&\triangleq &\inf_{\boldsymbol\alpha^+} S_A(\boldsymbol\alpha) \\
&\text{s.t. }& r_{l}>\sum_{i=1}^{M_*}[(\xi_{l-1}
-\alpha_i)^+-(\xi_{l}-\alpha_i)^+]. 
\end{IEEEeqnarray}

 Without loss of generality, consider the multiplexing gain $r_l$ given by
$r_l=k(\xi_{l-1}-\xi_{l})+\delta_l$, where $k\in[0,1,..., M_*-1]$
and $0\leq \delta_l<\xi_{l-1}-\xi_l$. Then, the infimum for
(\ref{eq:succDec1}) is found as
\begin{IEEEeqnarray}{lCl}\label{eq:SucceDecDivGain}
d_{ds}(r_l,\xi_{l-1},\xi_{l})=\Phi_k\xi_{l-1}-\Upsilon_k\delta_l,
\end{IEEEeqnarray}
with
\begin{IEEEeqnarray}{rCl}
\alpha_i^*=\begin{cases} \xi_{l-1},& 1\leq i< M_*-k,\\
\xi_{l-1}-\delta_l,& i=M_*-k,\\
0,& M_*-k<i\leq M_*.
\end{cases} 
\end{IEEEeqnarray}

Now, we solve (\ref{eq:MultilayerNBJD_DistExpProblemLayer}), using (\ref{eq:SucceDecDivGain}) for each layer, as a function of the power allocation $\xi_{l-1}$ and $\xi_{l}$, and the rate
$r_{l}$.

When no layer is successfully decoded, i.e., $l=0$, we have
\begin{IEEEeqnarray}{lCl}
\Delta^{bs}_0(\mathbf{r},\boldsymbol\xi)=\inf
(\nu-\beta)^++\beta+S_A(\boldsymbol\alpha) \\
\quad\text{s.t. }\,b\sum_{i=1}^{M_*}\left[(\xi_0- \alpha_i)^+-(\xi_1-\alpha_i)^+\right]<(br_1-(\nu-\beta))^+ .
\end{IEEEeqnarray}
The infimum is achieved by $\beta^*= \nu$ and using (\ref{eq:succDec1}), we have
\begin{IEEEeqnarray}{rCL}
\Delta^{bs}_0(\mathbf{r},\boldsymbol\xi)&=& \nu +d_{ds}\left(r_1,\xi_{0},\xi_{1}\right). 
\end{IEEEeqnarray}
At layer $l$, the distortion exponent  is given by the solution of the following optim
\begin{IEEEeqnarray}{rCl}
\Delta^{bs}_l(\mathbf{r},\boldsymbol\xi)&=&\inf\max\{b\bar{r}_1^l,(\nu-\beta)^+\}+\beta+S_A(\boldsymbol\alpha)  \\
&\text{s.t. }&
b\sum_{i=1}^{M_*}\left[\left(\xi_l-\alpha_i\right)^+-\left(\xi_{l+1}-\alpha_i\right)^+\right] 
<(b\bar{r}_1^{l+1}- \nu +\beta)^+-(b\bar{r}_1^{l}- \nu +\beta)^+. 
\end{IEEEeqnarray}

If $b\bar{r}_1^l\geq  \nu$, the infimum is obtained for $\beta^*=0$ and solving
\begin{IEEEeqnarray}{rCl}
\Delta^{bs}_l(\mathbf{r},\boldsymbol\xi)&=&\inf\max\{b\bar{r}_1^l, \nu \}+S_A(\boldsymbol\alpha)  \\
&\text{s.t. }&
\sum_{i^1}^{M_*}\left[\left(\xi_{l}-\alpha_i\right)^+-\left(\xi_{l+1}-\alpha_i\right)^+\right]<r_{k+1} .
\end{IEEEeqnarray}
Using (\ref{eq:succDec1}), we obtain the solutio as
\begin{IEEEeqnarray}{rCl}
\Delta^{bs}_l(\mathbf{r},\boldsymbol\xi)=\max\{x,b\bar{r}_{1}^l\}+d_{ds}\left(r_{l+1},\xi_{l},\xi_{l+1}\right). 
\end{IEEEeqnarray}

If $b\bar{r}_1^l\leq  \nu$, the infimum is given by $\beta^*=(x-b\bar{r}_1^l)^+$,
and again, we have a version of (\ref{eq:succDec1}) with the
distortion exponent
\begin{IEEEeqnarray}{rCl}
\Delta^{bs}_l(\mathbf{r},\boldsymbol\xi) =  \nu +d_{ds}\left(r_{l+1},\xi_{l},\xi_{l+1}\right). 
\end{IEEEeqnarray}

At layer $L$, the distortion exponent is the solution to the optimization problem
\begin{IEEEeqnarray}{rCl}
\Delta^{bs}_L(\mathbf{r},\boldsymbol\xi)&=&\inf
\max\left\{b\bar{r}_1^L,(\nu-\beta)^+\right\}+\beta+S_A(\boldsymbol\alpha) \\
&\text{s.t.}&b\sum_{i=1}^{M_*}[\left(\xi_{L-1}-\alpha_i\right)^+-\left(\xi_{L}-\alpha_i\right)^+] 
\geq\left(b\bar{r}_1^L-(\nu-\beta)\right)^+-\left(b\bar{r}^{L-1}_{1}-(\nu-\beta)\right)^+. 
\end{IEEEeqnarray}
The infimum is achieved by $\boldsymbol\alpha^*\!=\!0$ and $\beta^*\!=\!0$, and  is given
by
\begin{IEEEeqnarray}{rCl}
\Delta^{bs}_L(\mathbf{r},\boldsymbol\xi)&=&\max\left\{b\bar{r}_1^L,x\right\}, \quad \text{for } r_L\leq M_*(\xi_{L-1}-\xi_L). 
\end{IEEEeqnarray}
Note that the condition on $r_L$ always holds.

 Gathering all the results, the
distortion exponent problem in (\ref{eq:MultilayerNBJD_DistExpProblem}) is solved as the minimum of the exponent of each layer, $\Delta^{bs}_l(\mathbf{r},\boldsymbol\xi)$, which can be formulated as
\begin{IEEEeqnarray}{rCl}\label{eq:MLSystemMIn}
\Delta_{bs}^L(b,\nu)&=&\max_{\mathbf{r},\boldsymbol\xi} \,t \\
\text{s.t. }&&t\leq  \nu +d_{sd}\left(r_{1},\xi_{0},\xi_{1}\right), \\
&&t\leq \max\{b\bar{r}_1^l, \nu \}+d_{sd}\left(r_{l+1},\xi_{l},\xi_{l+1}\right),  \quad\text{for }l=1,\dots, L-1, \\
&&t\leq\max\{ b\bar{r}_1^L, \nu \}.
\end{IEEEeqnarray}

If $\nu  \geq b\bar{r}^L_1$, then $\max\{x,b\bar{r}_1^l\}= \nu$ for all $l$,
and the minimum distortion exponent is given by $\Delta^{bs}_{L}(\mathbf{r},\boldsymbol\xi)= \nu$, which implies $\Delta_{mj}^L(b,\nu)= \nu$. If
$\nu \leq  br_1$, then $\max\{x,b\bar{r}_1^l\}=b\bar{r}_{1}^l$ for all $l$. In general, if  $b\bar{r}^{q}_1< x \leq b\bar{r}^{q+1}_1$, $q=0,..., L$, and $\bar{r}^{0}_1\triangleq0$, $\bar{r}^{L+1}_1\triangleq\infty$, then (\ref{eq:MLSystemMIn}) can be formulated, using $r_l=k(\xi_{l-1}-\xi_{l})+\delta_l$, $\boldsymbol\delta\triangleq [\delta_1,\cdots,\delta_L]$ and $\boldsymbol\xi$, as the following linear optimization program:
\begin{IEEEeqnarray}{rCl}\label{eq:MlSystem2}
\Delta_{bs}^L(b,\nu)&=&\min_{\substack{1\leq q \leq L,\\0\leq k\leq M_*-1. }}\min_{\boldsymbol\delta,\boldsymbol\xi} \,-t \\
\text{s.t. }&&t\leq  \nu +\Phi_k\xi_{0}-\Upsilon_k\delta_1, \\
&&t\leq  \nu +\Phi_k\xi_{l}-\Upsilon_k\delta_{l+1}, \quad\text{for }\;l=1,\dots, q, \\
&&t\leq b\sum_{i=1}^l[k(\xi_{i-1}-\xi_{i})+\delta_i]+\Phi_k\xi_{l}-\Upsilon_k\delta_{l+1}, \\ &&\quad\text{for }\;l=q,\dots, L-1, \\
&&t\leq b\sum_{i=1}^L[k(\xi_{i-1}-\xi_{i})+\delta_i], \\
&&0\leq \delta_l< \xi_{l-1}-\xi_{l},\quad\text{for }\;l=1,\dots, L, \\
&&0\leq \xi_L\leq...\leq \xi_1\leq \xi_0=1, \\
&&\sum_{l=1}^{l'}[bk(\xi_{l-1}-\xi_{l})+\delta_l]<\nu.
\end{IEEEeqnarray}

The linear program (\ref{eq:MlSystem2}) can be efficiently solved using numerical methods. In Figure \ref{fig:MIMODistortionExponentMultiBad}, the numerical solution is shown. However, in the following we provide a suboptimal yet more compact analytical solution by fixing the multiplexing gains $\mathbf{r}$.  We fix the multiplexing gains  as $\hat{r}_l=[(k+1)(\xi_{l-1}-\xi_{l})-\epsilon_1]$, $\epsilon_1>0$
for $k=0,...,M_*-1$,  and $\delta_l\triangleq(\xi_{l-1}-\xi_{l})-\epsilon_1$, when the bandwidth ratio satisfies
\begin{IEEEeqnarray}{rCl}\label{eq:ML_Regime}
b&\in&\left[\frac{\Phi_{k+1}+\nu }{k+1},\frac{\Phi_k+\nu }{k} \right).
\end{IEEEeqnarray}

Assume $br_1\!\geq\!   \nu$. Then, each distortion exponent is found as
\begin{IEEEeqnarray}{rCl}\label{eq:BSLD-System}
\hat{\Delta}^{bs}_{0}(\mathbf{r},\boldsymbol\xi)&=& \nu +\Phi_k\xi_{l}-\Upsilon_k\delta_{l+1}, \\
\hat{\Delta}^{bs}_{l}(\mathbf{r},\boldsymbol\xi)&=&b\bar{r}_1^l+\Phi_k\xi_{l}-\Upsilon_k\delta_{l+1},\quad \text{for } l=1,...,L-1, \\
\hat{\Delta}^{bs}_{L}(\mathbf{r},\boldsymbol\xi)&=&b\bar{r}_1^L.
\end{IEEEeqnarray}

Similarly to the other schemes, for which the distortion exponent is maximized by equating the exponents, we look for the power allocation $\boldsymbol\xi$, such that all distortion exponent terms $\Delta^{bs}_{l}(\hat{\mathbf{r}},\boldsymbol\xi)$ in (\ref{eq:MultilayerNBJD_DistExpProblem}) are equal.

Equating all distortion
exponents $\hat{\Delta}^{bs}_l(\mathbf{\hat{r}},\boldsymbol\xi)$ for $l=2,...,L-1$, i.e.,
$\hat{\Delta}^{bs}_{l-1}(\mathbf{\hat{r}},\boldsymbol\xi)=\hat{\Delta}^{bs}_{l}(\mathbf{\hat{r}},\boldsymbol\xi)$, we have
\begin{IEEEeqnarray}{rCl}\label{eq:Sim}
d_{sd}\left(\hat{r}_{l},\xi_{l-1},\xi_{l}\right)=
br_{l}+d_{sd}\left(\hat{r}_{l+1},\xi_{l},\xi_{l+1}\right).
\end{IEEEeqnarray}
Since
$\hat{r}_l=[(k+1)(\xi_{l-1}-\xi_{l})-\epsilon_1]$, we have
\begin{IEEEeqnarray}{lCl}
d_{sd}\left(\hat{r}_{l},\xi_{l-1},\xi_{l}\right)&=&
\Phi_k\xi_{l-1}-\Upsilon_k(\xi_{l-1}-\xi_l-\epsilon_1). 
\end{IEEEeqnarray}
Substituting in (\ref{eq:Sim}), we find that the power allocations for $l\geq2$
need to satisfy,
\begin{IEEEeqnarray}{rCl}
(\xi_{l}-\xi_{l+1})=\eta_{k}(\xi_{l-1}-\xi_{l})+\mathcal{O}(\epsilon_1) ,
\end{IEEEeqnarray}
where $\eta_k$ is defined in (\ref{eq:ML_eta}) and $\mathcal{O}(\epsilon_1)$ denotes a term that tends to 0 as $\epsilon_1\rightarrow 0$. Then, for $l=2,...,L-1$ we obtain
\begin{IEEEeqnarray}{rCl}\label{eq:PowAllocRecurs}
\xi_l-\xi_{l+1}=\eta_k^{l-1}(\xi_1-\xi_2)+\mathcal{O}(\epsilon_1),
\end{IEEEeqnarray}
and $\xi_l$ can be found as
\begin{IEEEeqnarray}{rCl}
1-\xi_l&=&(1-\xi_1)+\sum_{i=1}^{l-1}(\xi_i-\xi_{i+1})+\mathcal{O}(\epsilon_1) \\
&=&(1-\xi_1)+\sum_{i=1}^{l-1}\eta_k^{i-1}(\xi_1-\xi_{2})+\mathcal{O}(\epsilon_1) \\
&=&(1-\xi_1)+(\xi_1-\xi_{2})\frac{1-\eta_k^{l-1}}{1-\eta_k}+\mathcal{O}(\epsilon_1) .
\end{IEEEeqnarray}

Then, for $l=2,...,L$, we have
\begin{IEEEeqnarray}{rCl}\label{eq:ML_gamma_l}
\xi_l=\xi_1-(\xi_1-\xi_2)\frac{1-\eta_{k}^{l-1}}{1-\eta_k}+\mathcal{O}(\epsilon_1).
\end{IEEEeqnarray}

From $\hat{\Delta}^{bs}_L(\mathbf{\hat{r}},\boldsymbol\xi)\!=\!b\bar{r}_1^L=b\sum_{i=1}^L(k+1)(\xi_{i-1}-\xi_i)$,
we have
\begin{IEEEeqnarray}{rCl}\label{eq:BSLD_DeltaL}
\hat{\Delta}^{bs}_L(\mathbf{\hat{r}},\boldsymbol\xi) 
&=&b(k+1)(\xi_0-\xi_1)+ b(k+1)
(\xi_2-\xi_1)\sum_{i=1}^{L}\eta_k^{i-1}+\mathcal{O}(\epsilon_1) \\
&=&b(k+1)\left[(\xi_0-\xi_1)+
(\xi_2-\xi_1)\frac{1-\eta_k^{L-1}}{1-\eta_k}\right]\!+\!\mathcal{O}(\epsilon_1).
\end{IEEEeqnarray}

Putting all together, from (\ref{eq:BSLD-System}) we obtain
\begin{IEEEeqnarray}{rCl}
\hat{\Delta}^{bs}_0(\mathbf{\hat{r}},\boldsymbol\xi) &=&  \nu + \Phi_k \xi_0-\Upsilon_k(\xi_0-\xi_1-\epsilon_1), \\
\hat{\Delta}^{bs}_l(\mathbf{\hat{r}},\boldsymbol\xi) &=& b(k+1)(\xi_0-\xi_1)+ \Phi_k \xi_1-\Upsilon_k(\xi_1-\xi_2+\epsilon_1), 
\quad \text{for } l=1,...,L-1, \\
\hat{\Delta}^{bs}_L(\mathbf{\hat{r}},\boldsymbol\xi) &=& b
(k+1)[(\xi_0-\xi_1)+(\xi_2-\xi_1)\Gamma_k]\!+\!\mathcal{O}(\epsilon_1).
\end{IEEEeqnarray}
By solving $\hat{\Delta}_{bs}^L(b,\nu)=\hat{\Delta}^{bs}_0(\mathbf{\hat{r}},\boldsymbol\xi)=\hat{\Delta}^{bs}_1(\mathbf{\hat{r}},\boldsymbol\xi)=\hat{\Delta}^{bs}_L(\mathbf{\hat{r}},\boldsymbol\xi)$, and letting $\epsilon_1\rightarrow 0$, we obtain (\ref{ML_distExpL}), and
\begin{IEEEeqnarray}{lCl}\label{ML_gamma1and2}
\xi_1&=&\frac{(\Upsilon_k+\Phi_k\Gamma_k)(\Upsilon_k+b(k+1)-\Phi_k- \nu )}{(\Upsilon_k+b(1+k))(\Upsilon_k+b(1+k)\Gamma_k)-b(k+1)\Phi_k\Gamma_k}, \\
\xi_1-\xi_2
&=&\frac{\Phi_k(\Upsilon_k+b(k+1)-\Phi_k- \nu )}{(\Upsilon_k+b(1+k))(\Upsilon_k+b(1+k)\Gamma_k)-b(k+1)\Phi_k\Gamma_k}.
\end{IEEEeqnarray}

For this solution to be feasible, the power allocation sequence has to satisfy $1\geq\xi_1\geq...\xi_L\geq0$, i.e., $\xi_l-\xi_{l+1}\geq0$. From (\ref{eq:PowAllocRecurs}) we need $\eta_k\geq0$ and $\xi_1-\xi_2\geq0$.
We have $\eta_k \geq0$ if
$b\geq\frac{\Phi_{k+1}}{k+1}$,
which holds in the regime characterized by (\ref{eq:ML_Regime}). Then, $\xi_1-\xi_2\geq0$ holds if $\Upsilon_k+b(k+1)-\Phi_k- \nu \geq 0$ and $(\Upsilon_k+b(1+k))(\Upsilon_k+b(1+k)\Gamma_k)-b(k+1)\Phi_k\Gamma_k\geq0$. It can be shown that
$(\Upsilon_k+b(1+k))(\Upsilon_k+b(1+k)\Gamma_k)-b(k+1)\Phi_k\Gamma_k$ is monotonically increasing in $b\geq0$, and positive for $k=0,...,M_*-1$. Therefore, we need to check if $\Upsilon_k+b(k+1)-\Phi_k- \nu \geq 0$. This holds since this condition is equivalent to
\begin{IEEEeqnarray}{rCl}
b\geq\frac{\Phi_{k+1}+\nu }{k+1}. 
\end{IEEEeqnarray}
Note that, in this regime, we have $\xi_1\geq0$. In addition, $\xi_l=\xi_1+(\xi_1-\xi_2)\Gamma_k\geq0$. Therefore, for each $k$ the power allocation is feasible in the regime characterized by (\ref{eq:ML_Regime}). It can also be checked that $br_1> \nu$ is satisfied. This completes the proof.



\bibliographystyle{ieeetran}
\bibliography{ref}
\end{appendices}

\end{document}